\documentclass[lettersize,journal]{IEEEtran}
\usepackage{amsmath,amssymb,amsfonts,amsthm,bm}
\usepackage{mathrsfs}
\usepackage{algorithm}
\usepackage{algpseudocode}

\newtheorem{lemma}{Lemma}

\usepackage{footmisc}
\usepackage{array}
\usepackage{multirow}
\usepackage{textcomp}
\usepackage{stfloats}
\usepackage[caption=false,font=normalsize,labelfont=sf,textfont=sf]{subfig}
\usepackage{url}
\usepackage{graphicx}
\usepackage{cite}
\usepackage{xcolor}
\usepackage{mathtools}
\usepackage{amsmath,amsfonts,bm}

\def\1{\bm{1}}

\def\rva{{\mathbf{a}}}

\def\rvf{{\mathbf{f}}}

\def\rvm{{\mathbf{m}}}

\def\rvp{{\mathbf{p}}}

\def\rvr{{\mathbf{r}}}

\def\rvx{{\mathbf{x}}}
\def\rvy{{\mathbf{y}}}
\def\rvz{{\mathbf{z}}}

\def\rmA{{\mathbf{A}}}

\def\rmH{{\mathbf{H}}}

\DeclareMathAlphabet{\mathsfit}{\encodingdefault}{\sfdefault}{m}{sl}
\SetMathAlphabet{\mathsfit}{bold}{\encodingdefault}{\sfdefault}{bx}{n}

\DeclareMathOperator*{\argmax}{arg\,max}

\begin{document}

\title{A Measurement Report Data-Driven Framework for Localized Statistical Channel Modeling}

\author{Xinyu~Qin, Ye~Xue, Qi~Yan, Shutao~Zhang, Bingsheng~Peng and Tsung-Hui~Chang~\IEEEmembership{Fellow,~IEEE}
\thanks{

The work was supported in part by the National Key Research and Development Program of China under Grant 2024YFA1014201 (Corresponding author: Tsung-Hui Chang).

Xinyu Qin and Bingsheng Peng are with the Shenzhen Research Institute of Big Data, School of Science and Engineering, The Chinese University of Hong Kong, Shenzhen, Guangdong 518172, China (email: xinyuqin@link.cuhk.edu.cn, bingshengpeng@link.cuhk.edu.cn). Ye Xue is with the Shenzhen Research Institute of Big Data, School of Data Science, The Chinese University of Hong Kong, Shenzhen, Guangdong 518172, China (e-mail: xueye@cuhk.edu.cn). Qi Yan and Shutao Zhang are with Networking and User Experience Lab, Huawei Technologies, China (email: yanqi1@huawei.com, zhangshutao2@huawei.com). Tsung-Hui Chang is with the Shenzhen Research Institute of Big Data, School of Artificial Intelligence, The Chinese University of Hong Kong, Shenzhen, Guangdong 518172, China (email: tsunghui.chang@ieee.org).}
}


\markboth{Journal of \LaTeX\ Class Files,~Vol.~14, No.~8, August~2021}%
{Shell \MakeLowercase{\textit{et al.}}: A Sample Article Using IEEEtran.cls for IEEE Journals}


\maketitle

\begin{abstract}
Localized statistical channel modeling (LSCM) is crucial for effective performance evaluation in digital twin-assisted network optimization. Solely relying on the multi-beam reference signal receiving power (RSRP), LSCM aims to model the localized statistical propagation environment by estimating the channel angular power spectrum (APS). However, existing methods rely heavily on drive test data with high collection costs and limited spatial coverage. In this paper, we propose a measurement report (MR) data-driven framework for LSCM, exploiting the low-cost and extensive collection of MR data. The framework comprises two novel modules. The MR localization module addresses the issue of missing locations in MR data by introducing a semi-supervised method based on hypergraph neural networks, which exploits multi-modal information via distance-aware hypergraph modeling and hypergraph convolution for location extraction. To enhance the computational efficiency and solution robustness, LSCM operates at the grid level. Compared to independently constructing geographically uniform grids and estimating channel APS, the joint grid construction and channel APS estimation module enhances robustness in complex environments with spatially non-uniform data by exploiting their correlation. This module alternately optimizes grid partitioning and APS estimation using clustering and improved sparse recovery for the ill-conditioned measurement matrix and incomplete observations. Through comprehensive experiments on a real-world MR dataset, we demonstrate the superior performance and robustness of our framework in localization and channel modeling.  
\end{abstract}


\section{Introduction}
\label{Section:Introduction}
\IEEEPARstart{W}{ith} the rapid evolution of wireless communications, network optimization has become increasingly critical for the development and deployment of next-generation wireless networks\cite{10155734,10636212,li2022real}. Millions of tunable network parameters are intricately coupled, requiring careful configuration to align with real-world conditions so as to ensure optimal network performance. Traditional optimization methods, relying on expert knowledge and trial-and-error campaigns, are labor-intensive and error-prone\cite{6353680}. Recent advancements in big data and cloud-based computational power have enabled digital twin-assisted network optimization \cite{10416388,li2025generative}, which creates virtual replicas of real-world wireless networks. These digital twins simulate network behaviors and optimize parameters using network simulators \cite{10605806,10697404,10234427}. However, to bridge the gap between simulated and real-world networks—and thereby ensure effective network optimization—it is essential to precisely characterize the localized propagation environment and accurately reproduce the intricate behavior of wireless channels between mobile devices and base stations (BSs).

To achieve these goals, localized statistical channel modeling (LSCM) was proposed in \cite{10299600}. Specifically, LSCM divides the service area of a BS into grids, collects multi-beam reference signal received power (RSRP) measurements within each grid, and estimates the channel angular power spectrum (APS) between the BS and the grid based on its statistical relationship with the multi-beam RSRP. The key concept is to utilize real-world measurements to extract the multi-path structure of the localized propagation environment, facilitating network performance evaluation under different parameter settings. However, as a data-driven framework, LSCM requires a large amount of real-world data to ensure its effectiveness and coverage. Existing methods\cite{10299600,9940390} depend on data collected from drive tests (DTs). Although DTs provide accurate measurements, they are time-consuming and labor-intensive, which restricts their practical application. Moreover, because DT data is typically collected along roads during planned time periods, its spatial coverage is inherently limited, restricting the coverage area and effective period of LSCM and potentially degrading its reliability\cite{8640805}.

In response to the limitations of DTs, measurement report (MR) data have emerged as a promising alternative for enhancing LSCM. When mobile devices initiate calls or access data services, MR data is generated to report connection states (e.g., RSRP) between mobile devices and BSs. Compared to using DT data, MR data-driven LSCM offers several advantages: 1) it enables large-scale network performance estimation by leveraging the ubiquitous presence of mobile devices; 2) it supports real-time and long-term monitoring through continuous data accumulation; and 3) it provides a cost-effective solution via automated data collection. However, the unique characteristics of MR data present the following challenges when applied to existing methods\cite{10299600,9940390,10924577}.

The first challenge is the lack of device location information in MR data due to privacy protection or mobile device settings. As previously mentioned, LSCM operates at the grid level using the in-grid average multi-beam RSRP as input. This not only improves the efficiency of large-scale channel modeling by reducing the storage and computational burden of massive data but also enhances the robustness of statistical channel modeling by alleviating the effects of measurement errors and data sparsity. For example, relying on the DT data, the existing methods \cite{10299600,9940390} construct equally sized geographic grids through uniform spatial discretization, exploiting the spatial correlation between adjacent grids in the propagation environment \cite{7275189}. However, this approach becomes infeasible when extended to MR data, as MR data only provide connection state information without high-accuracy location labels. Therefore, it is necessary to localize MR data. The 3GPP has introduced geometric localization methods based on time difference of arrival and angle of arrival \cite{fischer20215g}. Some data-driven methods employ machine learning and statistical analysis to extract localization patterns from RSRP, such as random forest \cite{zhu2016city} or deep neural networks \cite{butt2020rf}. However, their localization accuracy depends on sufficient fingerprint data, which is usually unavailable in practical scenarios.
 
The second challenge arises from the complex localized propagation environment and the non-uniform spatial distribution of MR data. On one hand, due to environmental factors such as terrain and buildings, physically consistent grids could have a dramatically different channel structure. As a result, grids construction that is solely based on geographic location may fail to align with the localized propagation environment. On the other hand, unlike DT data, which is densely and continuously collected through dedicated measurement campaigns in specified regions, MR data exhibits a non-uniform spatial distribution shaped by user mobility patterns\cite{8014487}, and some grids can have only a few RSRP samples, causing poor channel APS estimation. Therefore, the traditional uniform grid construction methods are no longer suitable. These factors motivate us to develop a more adaptive framework, where the channel APS estimation is jointly optimized with the grid construction to better capture the localized statistical channel structures in complex propagation environments with MR data.

To address the above challenges, this paper proposes an MR data-driven framework for localized statistical channel modeling (MR-LSCM). As illustrated in Fig.~\ref{fig:MR_LSCM}, the framework consists of two modules: 1) MR localization and 2) joint grid construction and channel APS estimation. Specifically, the MR localization module aims to localize extensive MR samples with minimal calibration fingerprints. To this end, we propose a hypergraph neural network-based semi-supervised MR localization method (HGNN-Loc). By leveraging the multi-modal information inherent in MR data, we model the data as a distance-aware hypergraph—a more general graph structure that encodes high-order data correlations with degree-free hyperedges\cite{bretto2013hypergraph}. Furthermore, the location information of MR data is extracted using hypergraph convolution\cite{gao2022hgnn+}. 

\begin{figure}[t]
    \centering
    \includegraphics[width=\linewidth]{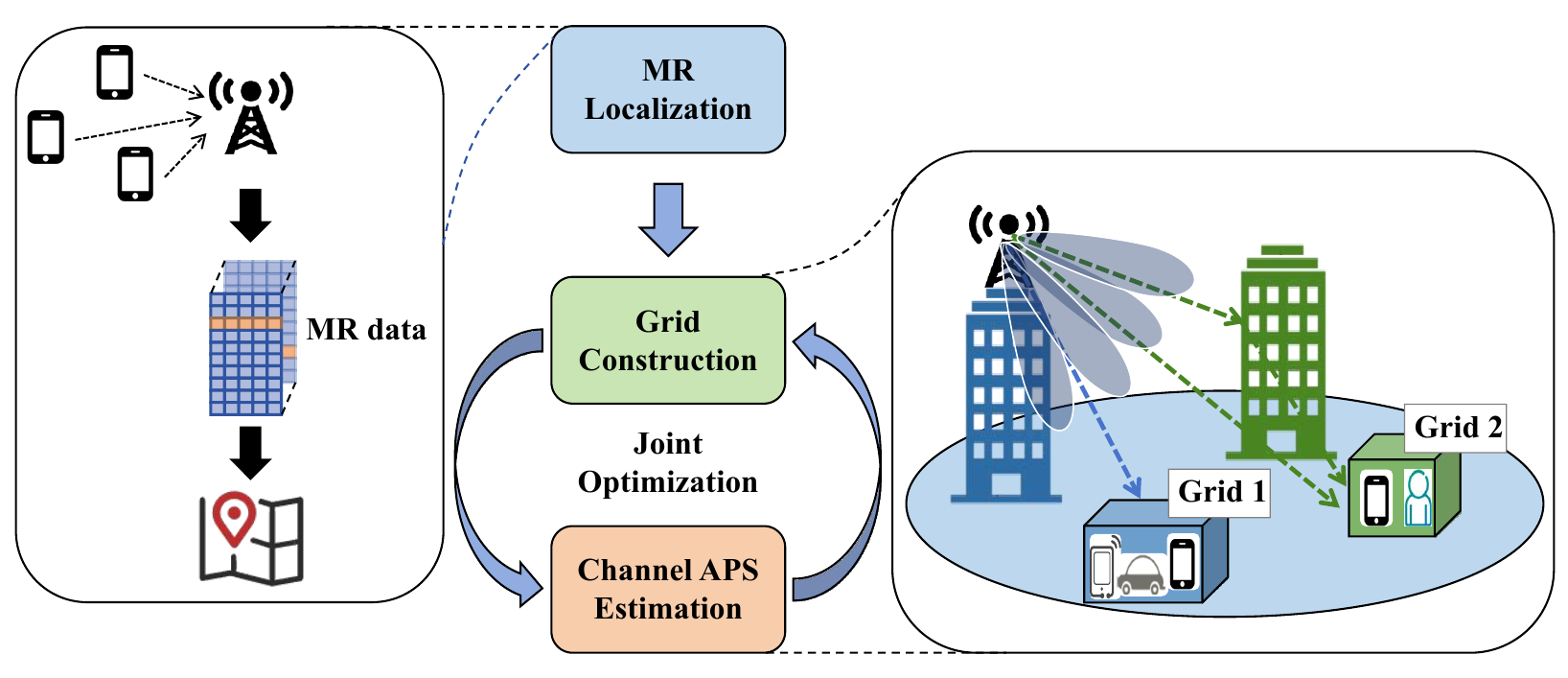}
    \caption{\textbf{The proposed MR-LSCM framework.} Our framework consists of two modules: 1) MR localization and 2) joint grid construction and channel APS estimation. The MR localization module predicts the location of each MR sample. Based on these predictions, the joint grid construction and channel APS estimation module discretizes the spatial domain and extracts multi-path channel structures of the localized propagation environment.}
    \label{fig:MR_LSCM}
    \vspace{-5pt}
\end{figure}

For the second module, we formulate a joint clustering and sparse recovery problem to unify grid construction and channel APS estimation. In particular, by clustering MR samples with similar channel APS, we construct grids in which samples share consistent channel multi-path structures, and estimate the channel APS within each grid via sparse recovery. To ensure that the constructed grids are well aligned with the localized propagation environment, we define a clustering space that integrates multi-beam RSRP and location. The physical consistency of the grids is quantified by computing the weighted sum of the in-grid RSRP variance and the location variance, which captures both the signal homogeneity and spatial correlation inherent in similar multi-path channel structures. Then, we design an alternating optimization approach to efficiently solve this non-convex problem with coupled variables. Considering the effects of the ill-conditioned measurement matrix and missing beam measurements, a geometry- and missing-value-aware non-negative orthogonal matching pursuit (GM-NNOMP) algorithm is proposed for channel APS estimation. Extensive experimental results on a real-world MR dataset demonstrate the superior performance and robustness of the proposed framework in localization and channel modeling.

\begin{table*}[t]
\centering
\caption{A detailed example of 5G MR data}
\begin{tabular}{ccc|cccccccccccccccc}
\hline
\multicolumn{3}{c|}{\textbf{Connection Information}}                                                   & \multicolumn{16}{c}{\textbf{Multi-beam RSRP (dBm)}}                                                                \\ \hline
\multicolumn{1}{c|}{\textbf{Time}}      & \multicolumn{1}{c|}{\textbf{gNodeBCellID}} & \textbf{CallID} & \multicolumn{8}{c|}{\textbf{Serving Cell}}                         & \multicolumn{8}{c}{\textbf{Neighboring Cell}} \\ \hline
\multicolumn{1}{c|}{2022/7/19 10:10:23} & \multicolumn{1}{c|}{12595215\_3}           & 268615725       & -73 & -80 & -71 & -73 & -80 & \text{\texttimes}   & -87 & \multicolumn{1}{c|}{-84} & -87   & \text{\texttimes}   & -89  & \text{\texttimes}  & \text{\texttimes}  & \text{\texttimes}  & \text{\texttimes}  & -89  \\
\multicolumn{1}{c|}{2022/7/19 10:10:24} & \multicolumn{1}{c|}{12595215\_3}           & 268615725       & -73 & -80 & -72 & -74 & -82 & -87 & \text{\texttimes}   & \multicolumn{1}{c|}{-86} & -87   & \text{\texttimes}   & -89  & \text{\texttimes}  & \text{\texttimes}  & \text{\texttimes}  & \text{\texttimes}  & -88  \\
\multicolumn{1}{c|}{2022/7/19 10:10:25} & \multicolumn{1}{c|}{12595215\_3}           & 268615725       & -78 & -79 & -70 & -74 & -82 & \text{\texttimes}   & -87 & \multicolumn{1}{c|}{-86} & -92   & \text{\texttimes}   & -95  & \text{\texttimes}  & \text{\texttimes}  & \text{\texttimes}  & \text{\texttimes}  & -87  \\
\multicolumn{1}{c|}{2022/7/19 10:10:26} & \multicolumn{1}{c|}{12595215\_3}           & 268615725       & -75 & -73 & -65 & -68 & -75 & \text{\texttimes}   & -82 & \multicolumn{1}{c|}{-78} & -89   & \text{\texttimes}   & \text{\texttimes}    & \text{\texttimes}  & \text{\texttimes}  & \text{\texttimes}  & \text{\texttimes}  & -80  \\ \hline
\end{tabular}
\label{tab:mr}
\vspace{-5pt}
\end{table*}

Our key contributions are summarized as follows.
\begin{itemize}
    \item \textbf{MR Data-driven LSCM:} We propose MR-LSCM, an MR data-driven framework for LSCM. To integrate MR data into LSCM, we design two modules: 1) MR localization module and 2) joint grid construction and channel APS estimation. The MR localization module predicts the locations of MR data. For joint grid construction and channel APS estimation, we formulate a joint clustering and sparse recovery problem to enhance solution robustness under complex environments with spatially non-uniform data. To construct grids aligned with the localized propagation environment, we define a clustering space that integrates multi-beam RSRP and location, jointly prompting the signal homogeneity and spatial correlation of physically consistent grids.
    \item \textbf{HGNN-based Semi-supervised MR Localization:} To address the lack of device location information in MR data, we propose HGNN-Loc. We exploit the multi-modal information inherent in MR data to construct a distance-aware hypergraph—a more general graph structure that enables modeling of high-order correlations through degree-free hyperedges. Furthermore, we extract the locations of each MR sample via hypergraph convolution. Through semi-supervised learning, HGNN-Loc can be effectively trained with minimal location labels.
    \item \textbf{Joint Grid Construction and Channel APS Estimation:} To solve the joint clustering and sparse recovery problem, we adopt an alternating optimization approach. Grids are constructed by clustering based on predicted locations and estimated channel APS. The average multi-beam RSRP within each grid is used to estimate the channel APS via sparse recovery. To enhance the robustness of sparse recovery under the ill-conditioned measurement matrix and incomplete observation, we design a GM-NNOMP algorithm by incorporating geometry-aware weighted column selection and missing-value-aware constrained non-negative least squares. 
\end{itemize}

The rest of this paper is organized as follows. Section II reviews the background and related work. The system model and problem formulation are presented in Section III. Section IV delineates the proposed MR-LSCM framework. Experimental results are demonstrated in Section V. Finally, Section VI concludes the paper.

\par \textit{Notations}: In this paper, scalars are denoted by lower-case letters $a$, vectors are denoted by boldface letters $\rva$, matrices are denoted by upper-case letters $\rmA$, and sets are denoted by $\mathcal{A}$. The cardinality of set $\mathcal{A}$ is denoted by $|\mathcal{A}|$. The $i$-th entry of a vector $\mathbf{a}$ is denoted by $[\mathbf{a}]_i$. The operators $(\cdot)^T$, $\|\cdot\|_l$ and $\odot$ represent transpose, $l$-th norm, and Hadamard product, respectively. $\mathbb{E}[\cdot]$ denote the statistical expectation.

\section{Background and Related Work}
In this section, we provide the background on 5G MR data, followed by a review of MR localization and LSCM.

\subsection{5G Measurement Report Data}
\label{Subsection: 5G MR}
MR data is generated to report connection states between mobile devices and connected cells during phone calls or data service access. Table~\ref{tab:mr} gives an example of 5G MR data reported by a mobile phone. The fundamental information required by our work consists of:

1) Connection information. BSs track the identities of mobile devices and gNodeBs\footnote{The gNodeB is a critical component in the 5G radio access network. It serves as the base station that communicates directly with mobile devices, facilitating wireless connectivity over the 5G network.}, including the serving cell and neighboring cells. The {\em gNodeBCellID} serves as both the gNodeB identifier and the cell identifier, uniquely identifying the connected cell. The {\em CallID} acts as a one-time service identifier to uniquely identify the connected mobile devices. Each MR sample records the reporting time through a timestamp. 2) Multi-beam RSRP. In 5G networks, the multi-beam RSRP can be measured from synchronization signals block (SSB) beams in downlink. However, due to the reporting mechanisms in \cite{3GPP_38.133_2022} and beam failures, there are missing values in multi-beam RSRP measurements, particularly those from neighboring cells.

\subsection{MR Localization}
Accurate localization of mobile devices is crucial for MR data applications, including network planning \cite{9456090}, resource allocation \cite{9449946}, and cellular traffic prediction \cite{8466626}. In both industry and academia, various localization methods with MR data have been developed. Google MyLocation\cite{googlemaps2020} approximates outdoor locations of MR data by the locations of connected cells. This simple method suffers from median errors of hundreds and even thousands of meters. More recently, data-driven methods have attracted intensive research interests, which can be categorized into fingerprint-based and learning-based localization. Fingerprint-based methods build a fingerprint database by collecting MR data  (e.g., timing advance, cell ID, and received signal strength indicator \cite{6062428,7524370,8057097}) tagged with corresponding locations. Localization is then performed through pattern matching with the known fingerprint database. With the development of wireless communication\cite{10158439}, such as larger signal bandwidths, denser network deployments, and multi-antenna technology, more types of fingerprints are supported to improve localization accuracy (e.g., multi-band and multi-cell RSRP\cite{10556755} and multi-beam RSRP\cite{10486831}). Learning-based methods establish a mapping from MR samples to associated locations by training models, including classic machine learning methods (e.g., random forest\cite{zhu2016city}) and deep learning methods (e.g., deep neural networks \cite{deeploc}). To achieve acceptable localization accuracy, these data-driven methods rely on training dataset constructed from a large number of labeled samples.

\subsection{Localized Statistical Channel Modeling}
\label{Background: LSCM}
Localized statistical channel modeling (LSCM) is a physics-based and data-driven channel modeling framework tailored for network optimization\cite{li2022real}. Unlike traditional geometry-based stochastic models that are limited to typical scenarios or deterministic channel modeling that relies on high-accuracy environmental information, LSCM effectively captures the statistical information of the localized propagation environment by estimating the channel APS from real-world RSRP measurements\cite{10299600}. Based on the derived statistical relationship between the beam-wise RSRP and the channel APS, the task of LSCM is formulated as a sparse recovery problem in\cite{10299600}. However, this problem is challenging to solve due to the ill-conditioned nature of the measurement matrix, characterized by non-uniform and closely parallel columns. To this end, \cite{10299600} proposed the weighted non-negative orthogonal matching pursuit (WNOMP) and second-order-statistics-based WNOMP algorithms. To leverage the spatial correlation of wireless channels, \cite{9940390} considered LSCM for multi-grid data and introduced a regularization term that constrains both sparsity and similarity. Similarly, \cite{10924577} proposed a graph neural network (GNN)-based approach for the multi-grid LSCM problem. However, these methods assume the availability of high-accuracy locations to construct uniform grids and neglect its correlation with channel APS estimation, which may degrade their performance when extended to MR data-driven LSCM.

\section{System Model and Problem Formulation}
\label{Section: System Model}
In this section, we first briefly review the localized statistical channel model. Subsequently, to develop MR data-driven LSCM, we formulate the joint grid construction and channel APS estimation problem. For ease of elaboration, we will assume that the locations of MR data are given, and leave the proposed MR localization method in Section~\ref{subsec:local}.

\subsection{Localized Statistical Channel Model}
\label{System:LSCM}
 Consider a wireless communication network where the BS is equipped with a uniform planar array (UPA) of $N_T = N_X \times N_Y$ antennas, where $N_X$ and $N_Y$ are the dimensions of the UPA. We discretize the tilt and azimuth AoDs uniformly into $N_V$ and $N_H$ angles based on the coordinate relative to the UPA, denoted as $\theta_i \  (1 \le i \le N_V)$ and $\varphi_j \ (1 \le j \le N_H)$, respectively. The downlink channel coefficient of the $(x, y)$-th antenna is modeled as\cite{38901}
\begin{align}
	h_{x, y} = & \sum_{i = 1}^{N_V} \sum_{j = 1}^{N_H}  \sqrt{\alpha_{i,j} } \times g_{i, j} \times e^{-j2\pi\frac{d_x x}{\lambda}\cos \theta_i\sin \varphi_j}  \nonumber\\
    &  \quad \qquad \times e^{-j2\pi\frac{d_y y }{\lambda}\sin \theta_i } \times e^{-j\omega_{i,j}-j\omega_{x,y}},
\end{align}
where $\alpha_{i,j}$ is the channel power of the path in the AoDs $(\theta_i,\varphi_j)$, $g_{i,j}$ is the antenna gain of the UPA, $\lambda$ is the wavelength of the carrier, $d_x$ and $d_y$ denote the spacing of adjacent antennas in the horizontal and vertical directions, respectively. $\omega_{i,j}$ and $\omega_{x,y}$ denote the random phase error between different angles and different antennas, respectively.

To establish the communication, directional reference signal beams are transmitted from the BS to sweep the angular space. Suppose that there are $M$ reference signal beams in total. The precoding matrix of the $m$-th  beam is denoted as $\boldsymbol{W}^{(m)} \in \mathbb{C}^{N_X \times N_Y}$ with the $(x,y)$-th entry $W_{x, y}^{(m)} = e^{j\phi_{x,y}^{(m)}}$, where $\phi_{x,y}^{(m)}$ is the phase offset of the $m$-th beam at the $(x,y)$-th antenna. Then, the RSRP of the $m$-th beam is given by
\begin{align}
	rsrp_m =  P\left|\sum_{x=1}^{N_X}\sum_{y=1}^{N_Y} h_{x, y} W_{x, y}^{(m)}\right|^{2},
\end{align}
where $P$ denotes the BS transmission power. 

Note that $rsrp_m$ contains three independent random variable, namely, $\omega_{i,j}$, $\omega_{x,y}$ and $\alpha_{i,j}$. The statistical relationship between $rsrp_m$ and $\alpha_{i,j}$ is established by taking expectation of $rsrp_m$ with respect to the random phase error $\omega_{i,j}$ and $\omega_{x,y}$. We define $\boldsymbol{\alpha} =[\alpha_{1,1},\dots,\alpha_{N_V,N_H}]^T \in \mathbb{R}^{N_A}$, where $N_A = N_V \times N_H$. Without loss of generality, we assume that the channel APS $\rvx = \mathbb{E}_T[\boldsymbol{\alpha}]$ remains invariant over $T$ successive samples due to the static environment. Let $\bm{r} = [rsrp_1,\dots,rsrp_M]^T \in \mathbb{R}^M$. Based on the assumption of random phase errors, the following statistical relationship between the multi-beam RSRP and channel APS can be derived in \cite{10299600}
\begin{align}\label{eq:y=Ax}
	{\bf y} = \mathbb{E}_T[\bm{r}] = \rmA {\bf x},
\end{align}
where $\rmA \in \mathbb{R}^{M \times N_A}$ is the measurement matrix representing the beamforming gain of the antenna array. In practice, LSCM is conducted at the grid level using the measurement mean as an approximation for the expected multi-beam RSRP ${\bf y}$ in the grid. This grid-based method not only enhance the computational efficiency but also the solution robustness.

Given the sparse nature of wireless channels due to the limited scattering environment \cite{10521790}, LSCM aims to solve the sparse recovery problem as follows
\begin{subequations} \label{eq:LSCM}
    \begin{align}
        \min_{\rvx} \quad & \|\rmA \rvx - \rvy\|_2^2 \label{eq:LSCM-obj} \\
        \text{s.t.} \quad & \|\rvx\|_0 \leq C, \label{eq:LSCM-c1} \\
        & \rvx \geq \mathbf{0}, \label{eq:LSCM-c2}
    \end{align}
\end{subequations}
where $C$ represents the maximum number of channel paths, and the constraint \eqref{eq:LSCM-c2} corresponds to the non-negativity of the power spectrum. It is noteworthy that the measurement matrix $\rmA$ is determined by the antenna parameters, while the channel APS $\rvx$ is determined by the localized propagation environment. LSCM aims to estimate the environmental statistical invariant $\rvx$ to facilitate network performance evaluation under different network parameters. However, as mentioned in Section~\ref{Section:Introduction} and Section~\ref{Background: LSCM}, problem~\eqref{eq:LSCM} is challenging to solve due to the ill-conditioned measurement matrix $\rmA$. In addition, the existing methods in \cite{10299600,9940390,10924577} rely on location information for grid construction and do not account for its coupling with channel APS estimation, making them inapplicable to the characteristics of MR data.

\begin{figure}[t]
\subfloat[Spatial Distribution of MR Data \label{fig:MR_distribution}]
{\centering
\includegraphics[width=\linewidth]{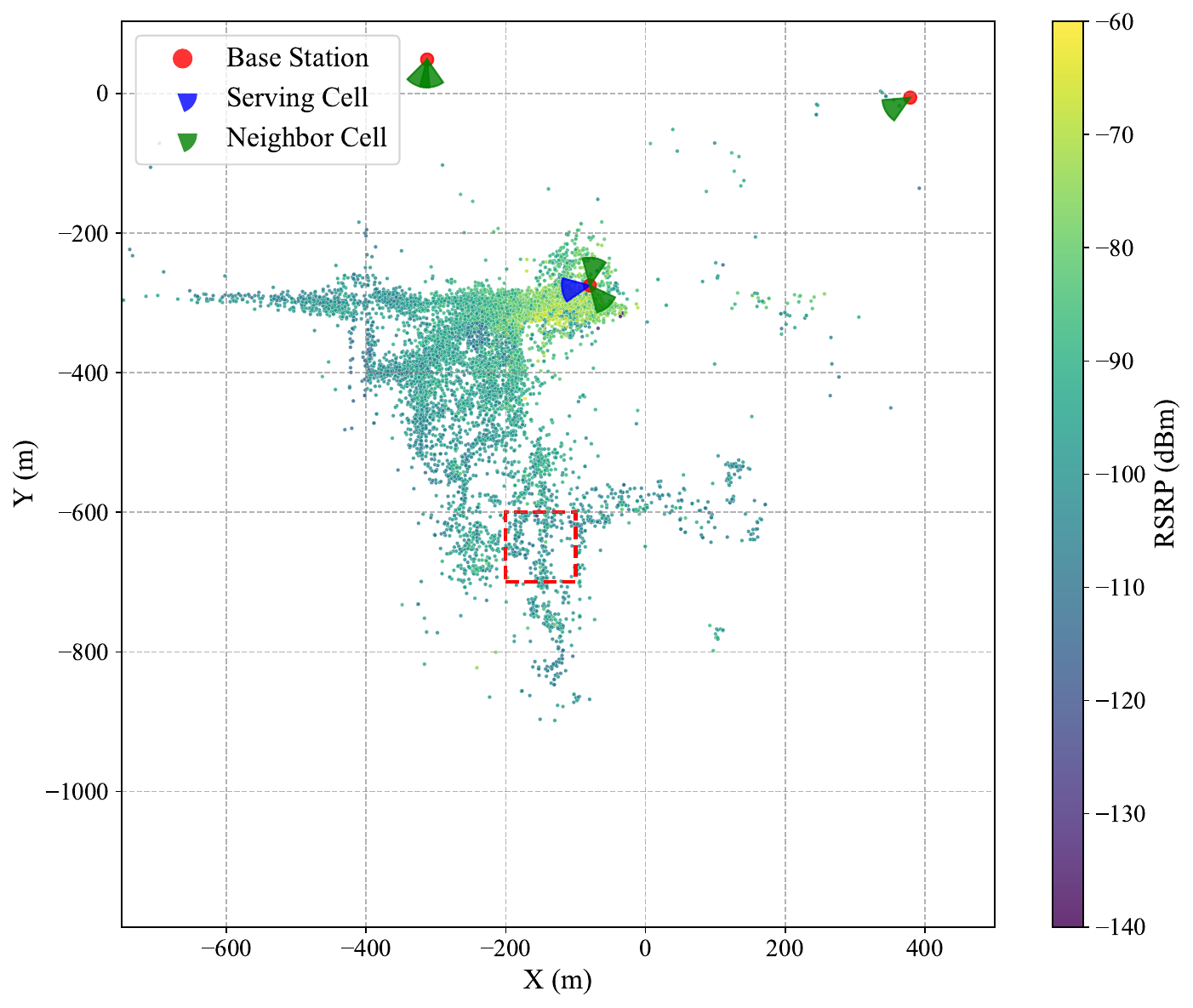}}
\\ 
\subfloat[Zoomed-in Region \label{fig:zoomed}]
{\centering
\includegraphics[width=0.51\linewidth]{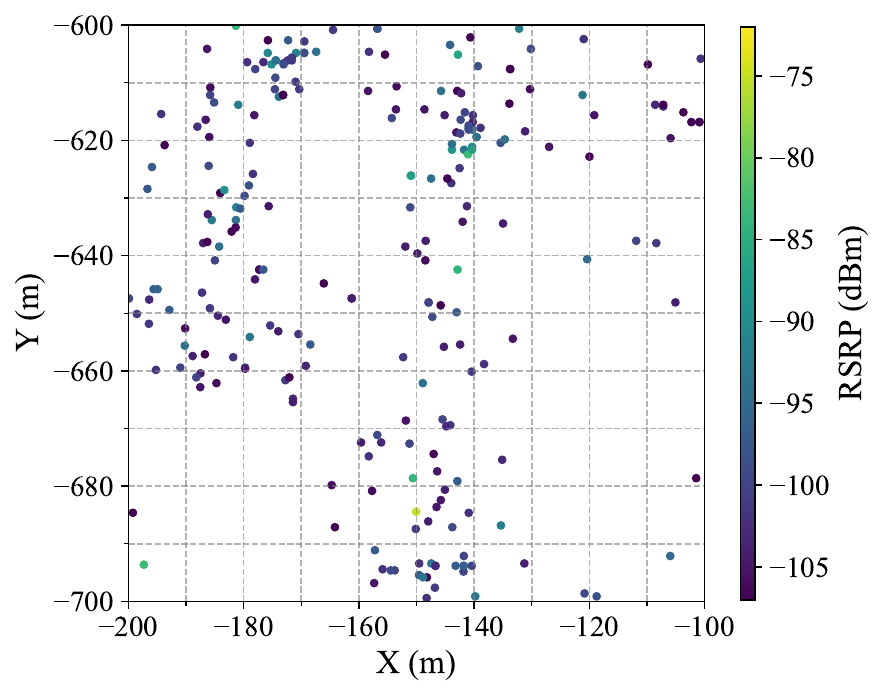}}
\subfloat[Uniform Grid \label{fig:uniform_grid}]
{\centering
\includegraphics[width=0.49\linewidth]{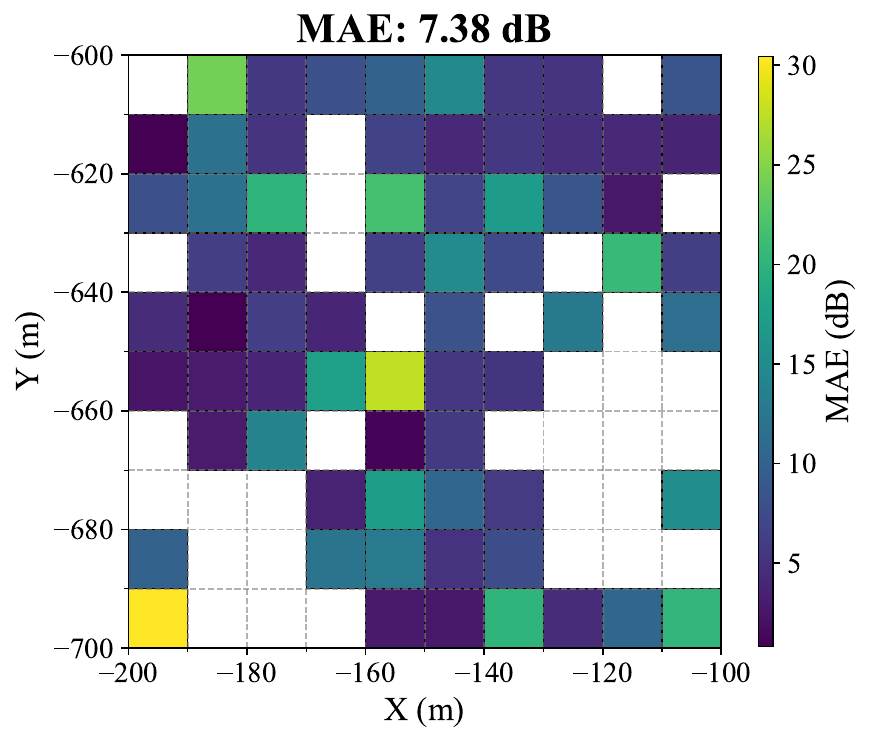}}
\caption{\textbf{Illustration of the spatial distribution of MR data, the zoomed region, and the corresponding uniform grid construction results.} Fig.~\ref{fig:MR_distribution} shows the spatial distribution of real-world MR data along with the corresponding RSRP values. Fig.~\ref{fig:zoomed} shows the zoomed-in region highlighted by the red dashed box in Fig.~\ref{fig:MR_distribution}. Fig.~\ref{fig:uniform_grid} shows the uniform grids ($10 \times 10~\mathrm{m}^2$ grids) within the zoomed region, as well as the MAE of RSRP prediction based on the estimated channel APS.}
\vspace{-5pt}
\end{figure}

\subsection{Problem Formulation}
To address the aforementioned challenges in MR data-driven LSCM, we tailor the problem formulation to the characteristics of MR data. For the lack of device location information in MR data, our first objective is to find a function that maps MR samples to their corresponding locations. We parameterize this function as $f_{\boldsymbol{\Xi}}$, where $\boldsymbol{\Xi}$ denotes the set of learnable parameters. Let $\rvp_i$ and $\hat{\rvp}_i \in \mathbb{R}^2$ denote the true and predicted locations of the $i$-th sample, respectively. Then, the MR localization problem can be formulated as
\begin{align}\label{eq:localization}
\boldsymbol{\Xi}^* = \arg\min_{\boldsymbol{\Xi}} \frac{1}{|\mathcal{P}|}\sum_{i\in\mathcal{P}} \| \hat{\rvp}_i - \rvp_i \|_2 ,
\end{align}
where $\mathcal{P}$ denotes the index set of MR samples with available location labels, and the mean distance error is used as the evaluation metric for localization accuracy. Given scarce fingerprint data (i.e., a small $\mathcal{P}$), a novel MR localization method will be proposed in Section~\ref{subsec:local}.

As previously mentioned, LSCM is conducted on grids, and typically employs uniform grids (e.g., $10 \times 10~\mathrm{m}^2$ grids). However, considering the complex propagation environment and the spatial non-uniformity of MR data, naive uniform grids solely based on geographic location may result in significant divergence of the channel characteristics within each grid\footnote{This divergence can be quantified by the mean absolute error (MAE) of RSRP prediction based on the estimated channel APS, calculated as $\text{MAE} = \frac{1}{M \times|\mathcal{G}|} \sum_{i \in \mathcal{G}} \left\| \mathbf{y}_i - \mathbf{A} \hat{\mathbf{x}}_i \right\|_1$, where $\rvy_i$ and $\rmA\hat{\rvx}$ denote the multi-beam RSRP of the serving cell in the $i$-th MR sample and the predicted value based on the estimated channel APS $\hat{\rvx}$, respectively, and $\mathcal{G}$ denotes the index set of MR samples within the grid.}. For example, Fig.~\ref{fig:MR_distribution} shows the spatial distribution of real-world MR data, while Fig.~\ref{fig:zoomed} presents the zoomed-in region highlighted by the red dashed box in Fig.~\ref{fig:MR_distribution}. It is evident that the MR data exhibits significant spatial non-uniformity, with some grids containing sparse samples. Fig.~\ref{fig:uniform_grid} displays the uniform grids within the zoomed-in region. It can be observed that grids with fewer samples tend to exhibit larger MAE values, which may degrade the performance of grid-based LSCM. Therefore, jointly considering both geographic location and channel structure in grid construction is essential for achieving improved performance in MR data-driven LSCM.

Based on the above analysis, we formulate a joint clustering and sparse recovery problem to unify grid construction and channel APS estimation. Given the total number of grids $K$ and the MR dataset $\{ \rvy_i \}_{i=1}^{N}$ from a specific serving cell, where $N$ is the number of samples and $\rvy_i \in \mathbb{R}^M$ denotes the multi-beam RSRP measurement of the serving cell in the $i$-th MR sample, our objective is to assign nearby samples to grids while ensuring that the samples within each grid share similar channel APS. To this end, we formulate this as the following joint clustering and sparse recovery problem
\begin{subequations} \label{eq:MR-LSCM}
    \begin{align}
        \min _{\{\mathcal{G}_k\}, \{\rvx_k\}} & \sum_{k=1}^{K} \frac{1}{|\mathcal{G}_k|} \sum_{i \in \mathcal{G}_k} \left( \| \rmA\rvx_k - \rvy_i \|_2^2 + \beta \| \bar{\rvp}_k - \hat{\rvp}_i \|_2^2 \right) \label{eq:MR-LSCM-obj} \\
        \text{s.t.} \quad & \ \mathcal{G}_k  \cap \mathcal{G}_j = \emptyset, \ \forall \; k,j \in  \left[ K \right], k \neq j, \label{eq:MR-LSCM-c1}  \\
        & \ \bigcup_{k=1}^K \mathcal{G}_k = \left[ N \right], \label{eq:MR-LSCM-c2} \\
        & \ \bar{\rvp}_k = \frac{1}{|\mathcal{G}_k|} \sum_{i \in \mathcal{G}_k}  \hat{\rvp}_i, \forall \;  k \in \left[ K \right], \label{eq:MR-LSCM-c3} \\
        & \ \| \rvx_k \|_0 \leq C, \forall \;  k \in \left[ K \right], \label{eq:MR-LSCM-c4} \\
        & \ \rvx_k \geq \mathbf{0},  \forall \;  k \in \left[ K \right],  \label{eq:MR-LSCM-c5}
    \end{align}    
\end{subequations}
In \eqref{eq:MR-LSCM}, $\mathcal{G}_k$ denotes the index set of MR samples within the grid $k$, and $\bar{\rvp}_k$ denotes the location centroid of the grid $k$ as shown in \eqref{eq:MR-LSCM-c3}. The constraints \eqref{eq:MR-LSCM-c1} and \eqref{eq:MR-LSCM-c2} represent the exclusivity and completeness of clustering, respectively. In the objective function \eqref{eq:MR-LSCM-obj}, the first term is from \eqref{eq:y=Ax}, it enforces the samples in the grid $k$ to have a common channel APS $\rvx_k$, while the second term quantifies the sample location difference within each grid. In addition, $\beta$ serves as a regularization factor to balance the relative importance of these two terms. 

By jointly optimizing the grid sets $\{\mathcal{G}_k\}$ and channel APS $\{\rvx_k\}$ to minimize the weighted sum of these two terms, our problem formulation in \eqref{eq:MR-LSCM} provides a more holistic and robust solution for MR data-driven LSCM. However, problem \eqref{eq:MR-LSCM} is challenging to solve due to the non-convex constraints and strong coupling of the optimization variables. In Section \ref{Subsection: JGCCM}, we propose a holistic solution to address these challenges.

\section{Proposed MR-LSCM Framework}
\label{sec:MR-LSCM}

In this section, we introduce the proposed MR-LSCM framework to address the problem~\eqref{eq:localization} and  problem~\eqref{eq:MR-LSCM}. As illustrated in Fig.~\ref{fig:MR_LSCM}, the framework consists of two key modules. First, the MR localization module predicts the locations of MR samples. Subsequently, the joint grid construction and APS estimation module alternately optimizes the grid sets and channel APS to capture the localized statistical channel structures with MR data. The design and interaction of each module are detailed below.

\subsection{HGNN-based Semi-supervised MR Localization}
\label{subsec:local}

To address the localization challenges posed by limited calibration fingerprints, we propose a hypergraph neural network (HGNN)-based semi-supervised MR localization method (HGNN-Loc). Hypergraphs are natural extensions of graphs by allowing an edge to associate with any number of vertices \cite{bretto2013hypergraph}. In this way, a hypergraph excels at modeling high-order correlations, compared to the traditional graph structure with pairwise connections. This makes it ideal for modeling the complex spatial correlations in MR data. The HGNN-Loc comprises two procedures, i.e., distance-aware hypergraph modeling and hypergraph convolution, as shown in Fig.~\ref{fig:HGNN}. 

\begin{figure*}[t]
    \centering
    \includegraphics[width=\linewidth]{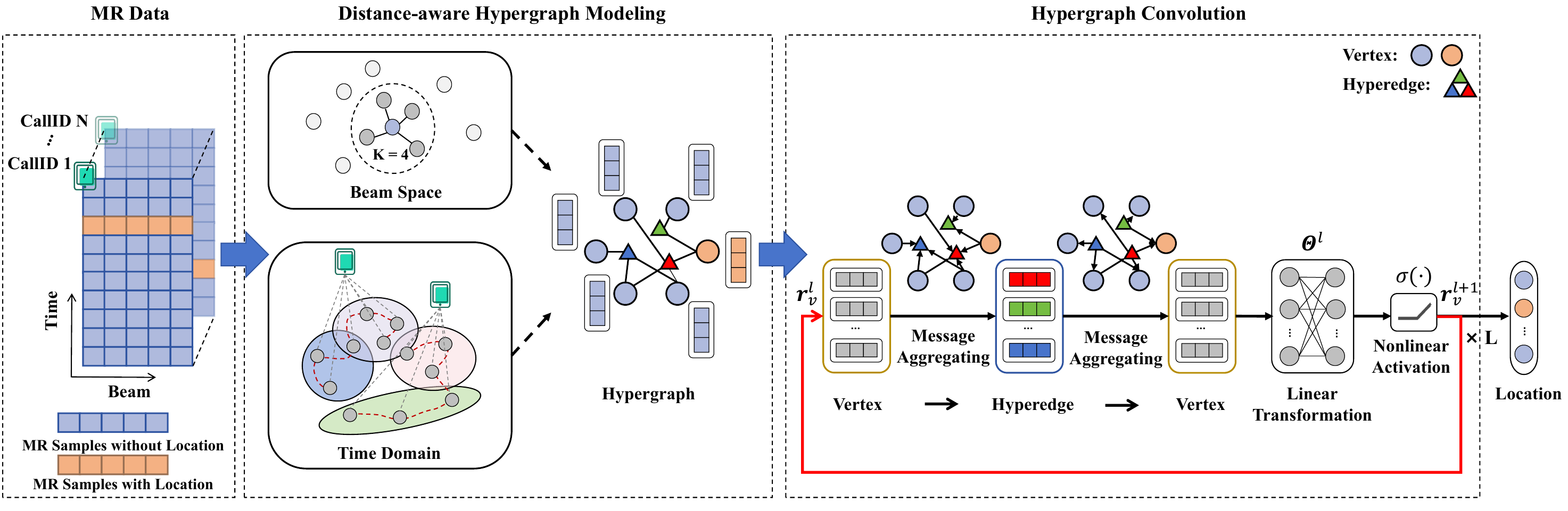}
    \caption{\textbf{The proposed HGNN-Loc.}
    In the left part, a distance-aware MR hypergraph is employed to model MR data.
    In the right part, the hypergraph convolution illustrates the message passing process via the vertex-hyperedge-vertex transformation and the red line from the output of the nonlinear activation function allows stacking multiple hypergraph convolution layers for deeper vertex embedding extraction.}
    \label{fig:HGNN}
    \vspace{-5pt}
\end{figure*}

\subsubsection{Preliminaries of Hypergraph} A hypergraph is defined as $\mathcal{H} = (\mathcal{V}, \mathcal{E})$, including a set of vertices $\mathcal{V}$ and a set of hyperedges $\mathcal{E}$. Each hyperedge $e \subseteq \mathcal{E}$ is associated with a non-empty subset of $\mathcal{V}$. In addition,  an incidence matrix $\rmH \in \mathbb{R}^{|\mathcal{V}| \times |\mathcal{E}|}$ is introduced to describe the structure of the hypergraph whose entries are defined as
\begin{align}\label{eq:consH}
    \rmH (v,e) = \begin{cases}
		1,  & \text{if} \; v \in e, \\
		0, & \text{otherwise}.
	\end{cases}
\end{align}

\subsubsection{Distance-aware Hypergraph Modeling} To capture the complex spatial correlations in MR data, it is essential to model a distance-aware hypergraph. We denote the distance-aware hypergraph as $\mathcal{H} = (\mathcal{V}, \mathcal{E})$ with $N$ vertices, where each vertex $v \in \mathcal{V}$ corresponds to an MR sample. Furthermore, to exploit the multi-modal information in the MR data, we introduce two types of hyperedges as follows.

\textbf{Hyperedge using beam space distance:}
This type of hyperedge  aims to connect spatially proximate vertices by exploiting the similarity of RSRP at adjacent locations. Given a vertex $v$ as the centroid, its $k$-nearest neighbors in the beam space can be connected by a hyperedge. We can define the set of all hyperedges for this type as
\begin{align}
    \mathcal{E}_1= \{N_{\text{beam}_k}(v) \;|\; \forall v \in \mathcal{V}\},
\end{align}
where $N_{\text{beam}_k}(v)$ denotes the $k$-nearest neighbors of the vertex $v$ measured by the beam space distance, which is calculated between vertices $v_i$ and $v_j$ as
\begin{align}
d(v_i, v_j) = \gamma \cdot \|\rvr_i - \rvr_j\|_2 + (1 - \gamma) \cdot \left(1 - \frac{\rvr_i \cdot \rvr_j}{\|\rvr_i\|_2 \|\rvr_j\|_2}\right),
\end{align}
where $\rvr_i = [\rvy_i; \rvy_i^n]$ denotes the multi-cell and multi-beam RSRP measurements\footnote{Note that in multi-beam RSRP measurements, missing values typically occur in weaker beams. For beams with missing values, we fill them with the minimum value (e.g., -140 dBm).} in the $i$-th sample. Here, $\rvy_i^{n} = [\rvy_i^1; \dots; \rvy_i^q; \dots; \rvy_i^Q]$ denotes the multi-beam RSRP measurements of the $Q$ neighboring cells in the $i$-th sample, with $\rvy_i^q \in \mathbb{R}^M$ representing the multi-beam RSRP measurements of the $q$-th neighboring cell. $\gamma$ is the weight coefficient. This weighted sum of the Euclidean and cosine distances jointly captures differences in both signal power and direction, enabling a more accurate assessment of spatial proximity.

\textbf{Hyperedge using {\em CallID} and timestamp:} 
Due to the limited mobility of mobile devices, temporally proximate MR samples from the same device are expected to be spatially adjacent. The same {\em CallID} identifies the same mobile device. Therefore, given the timestamp and {\em CallID} of a vertex $v$, all neighbors for this {\em CallID} within a time interval $\tau$ can be associated by a hyperedge. We can define the set of all hyperedges for this type as
\begin{align}
   \mathcal{E}_2 = \{N_{\text{time}_{\tau}}(v) \;|\; \forall v \in \mathcal{V} \},
\end{align}
where $N_{\text{time}_{\tau}}(v)$ denotes all neighbors of the vertex $v$ within the time interval $\tau$.

The left hand side of Fig.~\ref{fig:HGNN} illustrates the modeling of the distance-aware hypergraph. We construct two incidence matrices, $\rmH_1 \in \{0,1\}^{N \times |\mathcal{E}_1|}$ and $\rmH_2 \in \{0,1\}^{N \times |\mathcal{E}_2|}$, corresponding to the hyperedge sets $\mathcal{E}_1$ and $\mathcal{E}_2$, respectively, according to \eqref{eq:consH}. The overall hyperedge set is defined as $\mathcal{E} = \mathcal{E}_1 \uplus \mathcal{E}_2$, which is a multiset~\cite{blizard1989multiset}, where two hyperedges of different types associated with the same set of vertices are regarded as two different elements. By concatenating the two incidence matrices, we obtain the incidence matrix of the hypergraph $\rmH = [\rmH_1,\rmH_2]\in \{0,1\}^{N \times (|\mathcal{E}_1|+|\mathcal{E}_2|)}$.

\subsubsection{Hypergraph Convolution} Similar to graph learning, learning on a hypergraph can be viewed as the process of passing messages along the hypergraph structure. Inspired by the remarkable success in graph convolution, we employ hypergraph convolutions\cite{gao2022hgnn+} for location extraction.

Given a vertex $v$, we aggregate messages from the hyperedges containing $v$ to update its feature. Specifically, all hyperedges associated with $v$ can be represented by $\mathcal{N}_e(v)= \{e \in \mathcal{E} | \rmH(v,e) = 1\}$. To obtain the feature of each hyperedge $e \in \mathcal{N}_e(v)$, we aggregate messages from the vertices connected to $e$, i.e., $\mathcal{N}_{v}(e) = \{v \in \mathcal{V} \mid \mathbf{H}(v,e) = 1\}$. Let $\rvm_e^\ell$ and $\rvf_e^\ell \in \mathbb{R}^{d^\ell}$ denote the message and feature of hyperedge $e$ in the $\ell$-th layer, and let $\rvm_v^\ell$ and $\rvf_v^\ell \in \mathbb{R}^{d^\ell}$ denote the message and feature of vertex $v$ in the $\ell$-th layer, where $d^\ell$ is the feature dimension in the $\ell$-th layer. The hypergraph convolution in the $\ell$-th layer can be defined as

\begin{equation}\label{eq:HGNN}
    \left\{
    \begin{aligned}
        \rvm_e^\ell &= \sum_{v \in \mathcal{N}_v(e)} \frac{\rvf_v^\ell}{|\mathcal{N}_v(e)|}, \\
        \rvf_e^\ell &= w_e \cdot \rvm_e^\ell, \\
        \rvm_v^{\ell} &= \sum_{e \in \mathcal{N}_e(v)} \frac{\rvf_e^\ell}{|\mathcal{N}_e(v)|}, \\
      \rvf_v^{\ell+1} &= \sigma( \boldsymbol{\Theta}^\ell\rvm_v^{\ell}),
    \end{aligned}
    \right.
    \begin{array}{l}
        \left. \begin{aligned} \\ \\ \\ \end{aligned} \right\} \text{ vertex$\rightarrow$hyperedge} \\ \\ 
        \left. \begin{aligned} \\ \\ \\ \end{aligned} \right\} \text{ hyperedge$\rightarrow$vertex}
    \end{array}    
\end{equation}
where $w_e$ is the  weight of a hyperedge $e$ and is obtained by  \begin{align}
    w_e = \begin{cases}
		\operatorname{sigmoid}(w_1),  & \text{if} \; e \in \mathcal{E}_1, \\
		\operatorname{sigmoid}(w_2), & \text{otherwise}.
	\end{cases}
\end{align}
Learnable scalars $w_1$ and $w_2$ are used to differentiate the two types of modality information. $\boldsymbol{\Theta}^\ell \in \mathbb{R}^{d^{\ell+1} \times d^{\ell}}$ is a learnable linear transformation and $\sigma(\cdot)$ is a nonlinear activation function, which together are used to obtain the updated vertex feature $\rvf_v^{\ell+1}$, serving as the input feature for the next layer.

The right hand side of Fig.~\ref{fig:HGNN} illustrates the message passing process of hypergraph convolution via the vertex-hyperedge-vertex transformation. The model supports stacking $L$ hypergraph convolution layers for deeper vertex embeddings.

{\subsubsection{Semi-supervised Learning.} To make use of abundant unlabeled data with few location labels, we adopt a semi-supervised learning approach. For example, we adopt a 3-layer HGNN-Loc (i.e., $L=4$). For a vertex $v_i$, its feature in the first layer is the corresponding multi-cell and multi-beam RSRP $\rvf_{v_i}^1 = \rvr_i$, while the updated feature in the final layer represents its predicted location $\hat{\rvp}_i = \rvf_{v_i}^{4}$. All trainable parameters are $\boldsymbol{\Xi}=\{\boldsymbol{\Theta}^1, \boldsymbol{\Theta}^2, \boldsymbol{\Theta}^3, w_1, w_2\}$, and the loss function is articulated as
\begin{align}
	\text{LOSS}(\boldsymbol{\Xi}) = \frac{1}{|\mathcal{P}|} \sum_{i \in \mathcal{P}} \|\hat{\rvp}_{i}- \rvp_{i} \|_2^2.
\end{align}
Note that, in the training, only the predicted locations of MR samples with labels contribute to the computation of the loss function for backpropagation, thereby enabling the implementation of semi-supervised learning for MR localization.}

\subsection{Joint Grid Construction and Channel APS Estimation}
\label{Subsection: JGCCM}

Considering the non-convex constraints and strong coupling of the optimization variables in problem~\eqref{eq:MR-LSCM}, we employ an alternating optimization approach to jointly construct grids and estimate the channel APS.

\subsubsection{Grid Construction} For grid construction in (\ref{eq:MR-LSCM}), we fix $\{\rvx_k\}$ and optimize $\{\mathcal{G}_k\}$. The subproblem can be simplified as the following clustering problem:
\begin{subequations} \label{eq:sub1}
    \begin{align} 
        \min _{\{\mathcal{G}_k\}} & \sum_{k=1}^K  \frac{1}{|\mathcal{G}_k|}\sum_{i\in\mathcal{G}_k} \left( \| \rmA\rvx_k - \rvy_i \|_2^2 + \beta \| \bar{\rvp}_k - \hat{\rvp}_i \|_2^2 \right) \\
        \text{s.t.} \quad & \ \mathcal{G}_k  \cap \mathcal{G}_j = \emptyset, \ \forall \; k,j \in  \left[ K \right], k \neq j,   \\
        & \ \bigcup_{k=1}^K \mathcal{G}_k = \left[ N \right],  \\
        & \ \bar{\rvp}_k = \frac{1}{|\mathcal{G}_k|} \sum_{i \in \mathcal{G}_k}  \hat{\rvp}_i, \forall \;  k \in \left[ K \right]. \label{eq:sub1-c3}
    \end{align}  
\end{subequations}

Finding the global optimum of the minimum sum of squares clustering in problem~\eqref{eq:sub1} is NP-hard \cite{aloise2009np}. Therefore, we employ efficient heuristic algorithms to quickly converge to a local optimum. For instance, the classical $k$-means algorithm \cite{hartigan1979algorithm} iteratively updates cluster centers and assignments to group samples. In a similar fashion, we proceed with grid construction by alternating between the following two steps.

\textbf{Grid assignment step:} It is important to note that MR samples contain missing values; therefore, we define the valid distance as
\begin{align}\label{eq:valid_d}
    d_{\text{valid}}\left(\rvy_i,\rvx_k\right) = \frac{\operatorname{sum}(\rvm_i)}{M} \| \rvm_i \odot \left(\rmA\rvx_k - \rvy_i \right)\|_2,
\end{align}
where $\mathbf{m}_i \in \{{0, 1}\}^{M\times 1}$ is the mask vector associated with $\mathbf{y}_i$; specifically, $[\mathbf{m}_i]_j = 1$ indicates that the $j$-th beam value of $\mathbf{y}_i$ is valid, while $[\mathbf{m}_i]_j = 0$ otherwise. The scaling factor $\operatorname{sum}(\rvm_i)/M$ ensures appropriate normalization to account for the varying number of valid beams in each sample.

Given the cluster centers $\left\{ (\rmA\rvx_k, \bar{\rvp}_k) \right\}_{k=1}^{K}$, each sample can be assigned to the corresponding grid through a nearest-neighbor search, i.e.,
\begin{align} \label{eq:grid_assignment}
    g(i) = \arg\min_{k \in [K]} \; \left\{
 d_{\text{valid}}\left(\rvy_i,\rvx_k\right) + \beta  \| \bar{\rvp}_k - \hat{\rvp}_i \|_2 \right\} ,
\end{align}
where $g(i)$ denotes the grid assignment of the $i$-th sample. Thus, the index set of MR samples within the grid $k$ is given by $\mathcal{G}_k = \left\{ i \in [N]\; | \; g(i) = k \right\}$.

\textbf{Location centroid update step:} Update location centroids for samples assigned to each grid as in \eqref{eq:sub1-c3}.

\subsubsection{Channel APS Estimation.} For channel APS estimation in \eqref{eq:MR-LSCM}, we fix $\{\mathcal{G}_k\}$ and optimize $\{\rvx_k\}$, which leads to the following sparse recovery problem
\begin{subequations} \label{eq:sub2}
    \begin{align}
        \min_{\{\rvx_k\}} \quad & \sum_{k=1}^K \frac{1}{|\mathcal{G}_k|} \sum_{i\in\mathcal{G}_k}\|\rmA\rvx_k - \rvy_i\|_2^2  \\
        \text{s.t.} \quad 
        & \|\rvx_k\|_0 \leq C, \forall \;  k \in \left[ K \right], \\
        & \rvx_k \geq \mathbf{0}, \forall \;  k \in \left[ K \right].
    \end{align}
\end{subequations}

The following Lemma~\ref{lem1} shows that problem \eqref{eq:sub2} can be simplified to a problem like \eqref{eq:LSCM}. Solving the equivalent subproblem~\eqref{eq:sub2equ} enhances computational efficiency and solution stability through noise averaging and consistent sparsity.

\begin{lemma}\label{lem1}
Given matrix $\rmA$ and groups $\{\mathcal{G}_k\}_{k=1}^K$, problem (\ref{eq:sub2}) is equivalent to the following problem  
\begin{subequations} \label{eq:sub2equ}
    \begin{align}
        \min_{\{\rvx_k\}} \quad & \sum_{k=1}^K \left\|\rmA\rvx_k - \bar{\rvy}_k\right\|_2^2 \\
        \text{s.t.} \quad & \|\rvx_k\|_0 \leq C, \forall \;  k \in \left[ K \right], \\
        & \rvx_k \geq \mathbf{0}, \forall \;  k \in \left[ K \right],
    \end{align}
\end{subequations}
in the sense that the solution sets of (\ref{eq:sub2}) and (\ref{eq:sub2equ}) are identical, where \(\bar{\rvy}_k = \frac{1}{|\mathcal{G}_k|} \sum_{i \in \mathcal{G}_k} \rvy_i, \forall k \in [K]\).
\end{lemma}

\begin{proof}
Expand the objective of (\ref{eq:sub2}):
\begin{align}\label{eq:proof1}
& \sum_{k=1}^K \frac{1}{|\mathcal{G}_k|}\sum_{i\in\mathcal{G}_k} \|\rmA\rvx_k - \rvy_i\|_2^2  \nonumber \\
=& \sum_{k=1}^K \left[ \|\rmA\rvx_k\|_2^2 - 2\langle\rmA\rvx_k, \frac{1}{|\mathcal{G}_k|} \sum_{i \in \mathcal{G}_k} \rvy_i\rangle + \frac{1}{|\mathcal{G}_k|}\sum_{i\in\mathcal{G}_k}\|\rvy_i\|_2^2 \right] \nonumber \\
=& \sum_{k=1}^K \left[ \|\rmA\rvx_k - \bar{\rvy}_k\|_2^2 - \|\bar{\rvy}_k\|_2^2 + \frac{1}{|\mathcal{G}_k|} \sum_{i\in\mathcal{G}_k}\|\rvy_i\|_2^2 \right].
\end{align}
Since the last two terms in  (\ref{eq:proof1}) are irrelevant to $\rvx_k$, we finish the proof.
\end{proof}

As discussed in Section~\ref{Background: LSCM} for problem~\eqref{eq:LSCM}, subproblem~\eqref{eq:sub2equ} is particularly challenging to solve due to the ill-conditioned measurement matrix $\rmA$. Rather than adopting the algorithms in \cite{10299600}, we propose a geometry- and missing-value-aware NNOMP (GM-NNOMP) algorithm, as detailed in Algorithm~\ref{alg:GM-NNOMP}. To mitigate the impact of column magnitude\cite{10299600}, we define the column-wise normalized coefficient matrix $\hat{\rmA} = \left[ \hat{\rva}_1,\dots,\hat{\rva}_{N_A}\right]$ with $\hat{\rva}_i \triangleq \rva_i/\|\rva_i\|_2$ denoting the $i$-th column vector of the measurement matrix ${\rmA}$. Then, to mitigate the influence of ill-conditioned ${\rmA}$, the column selection in the OMP-based algorithm (Step 8 in Algorithm~\ref{alg:GM-NNOMP}) is crucial. For example, the WNOMP algorithm in \cite{10299600} adopts the following column selection criterion \begin{align}
    q = \underset{i \in \bar{\mathcal{S}}}{\argmax}\left( \frac{\hat{\rva}^T_i\rvr_{\mathcal{S}}}{\|\hat{\rmA}^T\rvr_{\mathcal{S}}\|_2} + \frac{\|\rva_i\|_2}{\sum_{i=1}^{N_A}\|\rva_i\|_2}\right).
\end{align}
This design aims to balance the effect of column magnitude and correlation with the residual. However, its underlying intuition—that channel paths are more likely to exist at angles with large beam gains—does not align with the fact that channel paths are determined by the propagation environment rather than the antenna array. Moreover, the issue of incomplete observations, i.e., missing values in multi-beam RSRP measurements $\rvy$, has not been considered, which may further degrade the accuracy of sparse recovery.

To address these shortcomings, we first introduce a geometry-aware weighted column selection criterion, as detailed in Step 8 of Algorithm~\ref{alg:GM-NNOMP}. Considering the high-dimensional channel angular domain, the geometric relationship between the BS and the grid can be utilized to pre-assign lower weights to some less likely candidate paths, thereby mitigating the effects of the ill-conditioned measurement matrix. Given the grid location centroid $\bar{\rvp}=(x, y)$ and the BS location $\rvp_{\text{bs}}=(x_{\text{bs}}, y_{\text{bs}}, h_{\text{bs}})$, where $h_{\text{bs}}$ denotes the height of the BS, we compute the relative angle $(\theta,\varphi)$ between the BS and the grid by Step 3 of Algorithm~\ref{alg:GM-NNOMP}. For the $i$-th candidate column, corresponding to the angle $(\theta_i, \varphi_i)$, its weight $w_i$ should be higher if it is closer to the relative angle $(\theta, \varphi)$, especially for grids under the line-of-sight (LoS) scenario \cite{zhang2016measurement}. We define the geometric angular weight $w_i$ based on the radial basis function (RBF) kernel shown in Step 4 of Algorithm~\ref{alg:GM-NNOMP}, where the kernel parameters $\sigma_{\theta}$ and $\sigma_{\varphi}$ control the considered angular spread. By multiplying the residual correlation by the geometric angular weight, the column selection process prioritizes channel paths that are not only highly correlated with the residual but also geometrically feasible.

\begin{algorithm}[t]
\caption{Geometry- and Missing-value-aware NNOMP (GM-NNOMP)}
\label{alg:GM-NNOMP}
\textbf{Input:} $\rvy$, $\rmA$, $\bar{\rvp}$, $\rvp_{\text{bs}}$, $C$, $\sigma_{\theta}$, $\sigma_{\varphi}$\\
\textbf{Output:} $\rvx$
\begin{algorithmic}[1]
\State \textbf{Preprocessing:}
\State $\hat{\rmA} = [\hat{\rva}_1, \dots, \hat{\rva}_{N_A}]$, $\hat{\rva}i \triangleq \rva_i / \|\rva_i\|_2$
\State $\theta = \arctan\left((y - y_{\text{bs}}) / (x - x_{bs})\right)$, $\varphi = \arctan\left(h_{\text{bs}} / \sqrt{(x-x_{\text{bs}})^2 + (y-y_{\text{bs}})^2}\right)$
\State $w_i = \exp\left(-((\theta_i - \theta)^2 / \sigma_{\theta}^2) - ((\varphi_i - \varphi)^2 / \sigma_{\varphi}^2)\right), \forall i \in [N_A]$
\State $\Omega = \{ i \in [M]\mid [\rvy]_i \text{ is not missing} \}$
\State \textbf{Initialization:} $\rvx = \mathbf{0}$, $\mathcal{S} = \emptyset$, $\rvr_{\mathcal{S}} = \rvy$
\Repeat
\State $q = \arg\max_{i \in \bar{\mathcal{S}}} (w_i \cdot \hat{\rva}_{i}^T \rvr_{\mathcal{S}} )$
\State $\mathcal{S} = \mathcal{S} \cup {q}$
\State $\mathbf{x}_\mathcal{S} = \arg\min_{\rvz \ge 0,, \rmA_{\bar{\Omega},\mathcal{S}} \rvz \le \rvy_{\min}} \|\rmA_{\Omega,\mathcal{S}} \rvz - \rvy_{\Omega}\|_2^2$
\State $\rvx_{\bar{\mathcal{S}}} = \mathbf{0}$
\State $\mathcal{S} \leftarrow \operatorname{supp}(\rvx)$
\State $\rvr_{\mathcal{S}} = \rvy - \rmA \mathbf{x}$
\Until{$\max_{i \in \bar{\mathcal{S}}} (\hat{\rva}_{i}^T \rvr_{\mathcal{S}}) < 0$ or $|\mathcal{S}| = C$}
\end{algorithmic}
\end{algorithm}

Once the support $\mathcal{S}$ is determined, the next step is to estimate the values of nonzero entries $\rvx_{\mathcal{S}}$. Considering missing values typically occur on weak beams, we construct a missing-value-aware constrained NNLS. Given the index set of valid beam $\Omega$ and its complement $\bar{\Omega}$, $\rvx_{\mathcal{S}}$ can be estimated by Step 10 of Algorithm~\ref{alg:GM-NNOMP}, where $\rmA_{\Omega,\mathcal{S}}$ represents the submatrix of $\rmA$ formed by selecting rows indexed by $\Omega$ and columns indexed by $\mathcal{S}$, $\rvy_{\Omega}$ represents the RSRP values of the valid beams, and $\rvy_{\text{min}} \in \mathbb{R}^{|\bar{\Omega}|}$ is a vector with all entries equal to the minimum RSRP value in $\rvy_{\Omega}$. By incorporating the linear inequality to constrain the recovered values of the missing beams, we can ensure that the estimated channel APS is more physically plausible and more robust to missing measurements.

\subsection{Overall implementation of MR-LSCM}

At this point, the aforementioned two modules collectively form the MR-LSCM framework. First, HGNN-Loc exploits the multi-modal information in MR data by modeling it as a distance-aware hypergraph. The hypergraph captures spatial correlations among MR samples through hyperedges constructed based on beam space distance, as well as hyperedges formed according to {\em CallID} and timestamp. Furthermore, the location of each MR sample is extracted via hypergraph convolution. HGNN-Loc is trained in a semi-supervised manner using only a limited number of location labels. 

After predicting the locations of all MR samples $\{\hat{\rvp}_i\}_{i=1}^N$, we employ Algorithm~\ref{alg:JGCCAE} to perform joint grid construction and channel APS estimation. Considering the spatial non-uniformity of MR data, the grid sets $\{\mathcal{G}_k\}_{k=1}^K$ are more effectively initialized using $k$-means++ \cite{arthur2006k} on $\{(\mathbf{y}_i, \hat{\mathbf{p}}_i)\}_{i=1}^N$. Subsequently, the channel APS ${\mathbf{x}_k}$ is estimated by Algorithm~\ref{alg:GM-NNOMP} using the average multi-beam RSRP and location centroid, while the grid sets ${\mathcal{G}_k}$ are clustered based on the estimated channel APS. Through alternating optimization, grid construction and channel APS estimation iteratively refine each other.

\begin{algorithm}[t]
\caption{Joint Grid Construction and Channel APS Estimation}\label{alg:JGCCAE}
\textbf{Input:} $\{\rvy_i\}_{i=1}^N$, $\{\hat{\rvp}_i\}_{i=1}^N$, $\rmA$, $\rvp_{\text{bs}}$, $C$, $\sigma_{\theta}$, $\sigma_{\varphi}$.\\
\textbf{Output:} $\{\mathcal{G}_k\}_{k=1}^K$, $\{\rvx_k\}_{k=1}^K$.\\
\textbf{Initialization:} initialize $\{\mathcal{G}_k\}_{k=1}^K$ by using $k$-means++ on $\{(\rvy_i,\hat{\rvp}_i)\}_{i=1}^N$.
\begin{algorithmic}[1]
\For{iteration $i = 0,1,2,\dots,I$}
    \For{grid $k = 0,1,2,\dots,K$}
        \State Compute $\bar{\rvy}_k$ based on $\mathcal{G}_k$.
        \State $\rvx_k \leftarrow \operatorname{GM-NNOMP}(\bar{\rvy}_k,\rmA,\bar{\rvp}_k,\rvp_{\text{bs}},C,\sigma_{\theta},\sigma_{\varphi})$.
    \EndFor
    \State Obtain $\{\mathcal{G}_k\}_{k=1}^K$ according to \eqref{eq:grid_assignment};
    \State Update $\{\bar{\rvp}_k\}_{k=1}^K$ according to \eqref{eq:sub1-c3}.
\EndFor
\end{algorithmic}
\end{algorithm}

\section{Experimental Results}
\label{Section:Results}
In this section, the effectiveness of the proposed MR-LSCM framework is evaluated through comprehensive experiments on a real-world MR dataset. We provide a detailed analysis of each module within the framework, including the HGNN-Loc for MR localization, the GM-NNOMP algorithm for channel APS estimation, and the joint grid construction and channel estimation algorithm. Finally, we evaluate the robustness of the overall framework against localization errors.

\subsection{Real-world MR Dataset}
The real-world MR dataset was gathered in a city of China, as depicted in Fig.~\ref{fig:MR_distribution}. The dataset includes a total of 12254 samples from a serving cell and 8 neighboring cells, covering an area of approximately 1.2 km$^2$. It contains MR samples collected both before and after network parameter adjustments for the serving cell. The adjusted parameters include the electronic downtilt and coverage scenario of the antenna array, where the coverage scenario is adjusted by selecting different codebooks for the SSB beams. We use the samples collected before the serving cell parameter adjustment as the training set, which includes 9810 samples, and those collected after the adjustment as the test set, which includes 2444 samples. In multi-beam RSRP measurements, up to 8 SSB beams can be measured for each cell. In general, the serving cell provides much more measured beams than the neighboring cells. We uniformly discretize the tilt angle and azimuth angle with intervals of $2^{\circ}$ and $5^{\circ}$, respectively, i.e., $\theta_i \in \{-90, -88, \dots, 88, 90\}$ and $\varphi_j \in \{-90, -85, \dots, 265\}$. Thus, $\theta_i$ and $\varphi_j$ have $N_V=91$ and $N_H=72$ possible values, respectively, with a total of $N_A = N_V \times N_H = 6552$ angles. To enable reliable performance evaluation, the MR data is labeled with Global Positioning System (GPS) coordinates to serve as ground truth. 

\begin{figure}[!t]
    \centering
    \includegraphics[width=\linewidth]{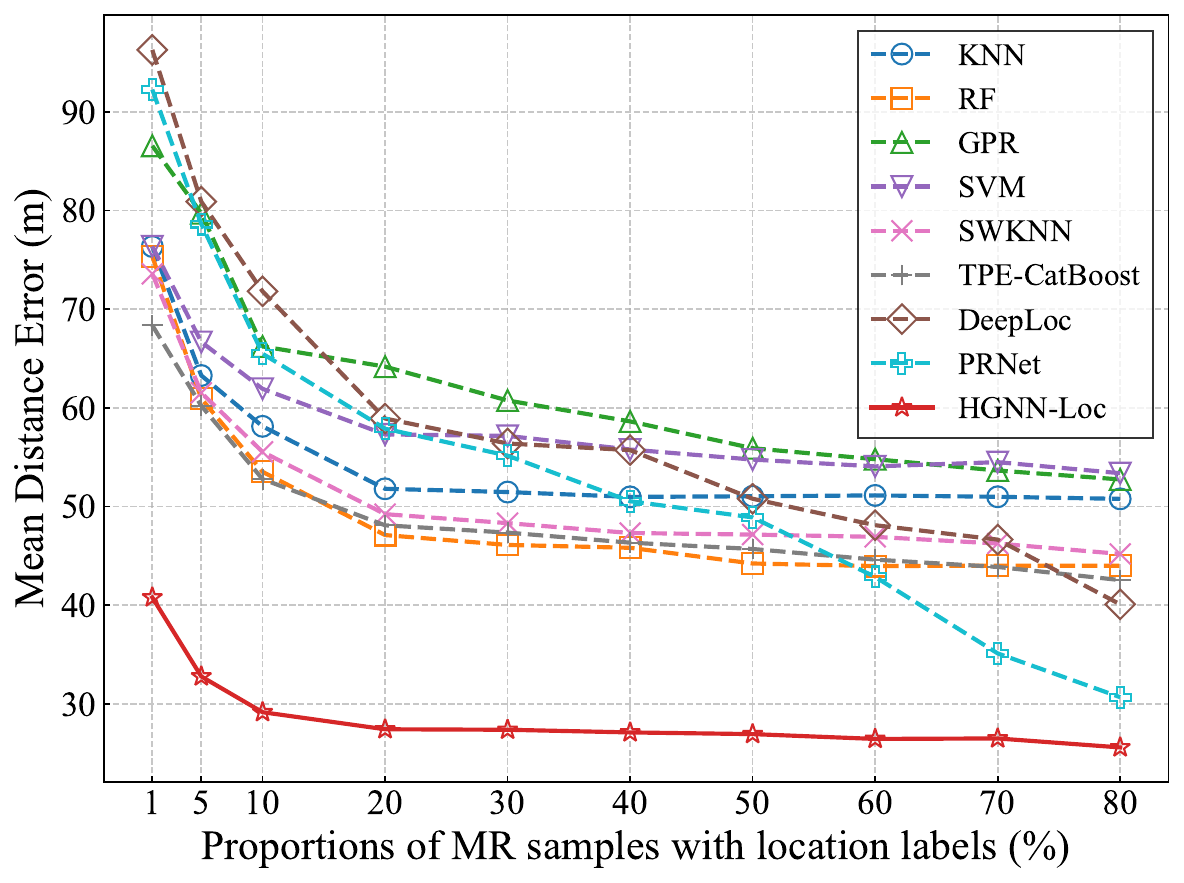}
    \caption{Mean distance error comparison for different localization methods under varying proportions of MR data with location labels.}
    \label{fig:LocalComp}
    \vspace{-5pt}
\end{figure}

\subsection{Performance Comparison of MR Localization}
We compare the proposed HGNN-Loc with state-of-the-art localization methods for performance benchmarking, categorized into machine learning-based and deep learning-based methods. The machine learning-based methods include KNN\cite{8430378}, random forest (RF)\cite{zhu2016city}, Gaussian process regression (GPR)\cite{KUMAR20161}, support vector machine (SVM)\cite{4079213}, SWKNN\cite{10556755} and TPE-CatBoost\cite{10486831}. The deep learning-based methods include DeepLoc\cite{deeploc} and  PRNet\cite{prnet}. The parameters of all methods are determined and tuned using ten-fold cross validation. For KNN, the optimal number of neighbors was found to be 5. A RBF kernel is used for both GPR and SVM. The SBM block size and the number of nearest neighbors for SWKNN are both set to 3. DeepLoc divides the geographical space into equally-sized grids, transforming the localization task into a classification task by predicting the grid of each sample. For DeepLoc, we use a grid size of 10 $m$ and a three-layer neural network. PRNet integrates convolutional neural networks, long short-term memory cells, and attention mechanisms to capture complex spatiotemporal correlations. The width of each layer in the proposed HGNN-Loc is $\{200,1000,2\}$. 

We assume that the location labels are available for only a certain proportion of the MR samples in the training set, and aim to predict the locations of the remaining MR samples within the training set. The proportion of labeled samples is varied from as low as 1\% to as high as 80\%. We evaluate localization accuracy with the mean distance error, defined as
\begin{align}
    \text{Mean Distance Error} = \frac{1}{|\mathcal{P}|}\sum_{i\in\mathcal{P}} \| \hat{\rvp}_i - \rvp_i \|_2 ,
\end{align}
where $\mathcal{P}$ denotes the index set of MR samples with location labels randomly selected from the training dataset according to the specified proportion.

As shown in Fig.~\ref{fig:LocalComp}, our proposed localization method consistently outperforms all other localization methods under different proportions of location labels. In particular, when location labels are extremely scarce, HGNN-Loc achieves a mean distance error of $40.79~m$ with only 1\% location labels. Although PRNet achieves a mean distance error of $30.68~m$ with 80\% location labels, its performance degrades significantly as the proportion of location labels decreases. Among two deep learning-based methods, both DeepLoc and PRNet exhibit much higher mean distance errors compared to HGNN-Loc. While PRNet integrates long short-term memory cells and attention mechanisms to capture contextual dependencies, its supervised learning-based training is prone to overfitting when only a small number of location labels are available. In contrast, HGNN-Loc demonstrates robust predictive capability even with limited calibration fingerprints, benefiting from the efficient inductive bias introduced by our distance-aware hypergraph modeling, which seamlessly fuses multi-modal information from MR data.

\begin{table}[H]
	\centering
	\caption{Test MAE Comparison of Different APS Estimation Algorithms}
        \begin{tabular}{|c|c|c|c|c|c|}
        \hline
                                                                 & LASSO & NNOMP & WNOMP & SWOMP & \begin{tabular}[c]{@{}c@{}}GM-\\ NNOMP\end{tabular} \\ \hline
        \begin{tabular}[c]{@{}c@{}}Test MAE \\ (dB)\end{tabular} & 8.44  & 7.76  & 7.34  & 7.41  & \textbf{6.53}                                         \\ \hline
        \end{tabular}
	\label{tb:test_mae}
\end{table}

\subsection{Performance Comparison of Channel APS Estimation}
In this subsection, we evaluate the effectiveness of our proposed channel APS estimation algorithm. The LASSO \cite{4839045}, NNOMP \cite{4655436}, WNOMP \cite{10299600}, and SWOMP \cite{10299600} algorithms are used as baselines for comparison. Since the ground-truth APS is unavailable, we evaluate performance using the mean absolute error (MAE) of RSRP prediction on the test set, defined as 
\begin{align}
    \text{Test MAE} = \frac{1}{K}\sum_{k=1}^K\frac{1}{M\times|\mathcal{G}'_k|}\sum_{i \in \mathcal{G}'_k}\| \rmA'\hat{\rvx}_k - \rvy'_i\|_1,
\end{align}
where $\hat{\rvx}_k$ denotes the estimated channel APS of the grid $k$, $\rmA'$ denotes the measurement matrix after network parameter adjustments, $\rvy'_i$ denotes the multi-beam RSRP of the $i$-th MR sample in the test set, $\rmA'\hat{\rvx}_k$ and $\rvy'_i$ are transformed into the dB domain. $\mathcal{G}'_k$ denotes the index set of MR samples within in the grid $k$ in the test set. To ensure geographic consistency of the same grids between the training and test sets, thereby guaranteeing the validity of the evaluation, we assign test samples to grids based on their true locations, defined as
\begin{align}
    \mathcal{G}'_k = \left\{i \in [N'] \mid k = \arg\min_{k \in [K]} \left(\min_{j \in \mathcal{G}_k}\|\rvp'_i - \rvp_j\|_2^2\right)\right\},
\end{align}
where $N'$ is the number of test samples and $\rvp'_i$ is the true location of the $i$-th MR sample in the test set. This assignment ensures that the training and test samples within the same grid are measured in the same geographical region.

To eliminate the impact of different grid constructions on channel APS estimation, all algorithms are implemented using uniform grids with a grid width of 10 meters. As shown in Table~\ref{tb:test_mae}, our proposed GM-NNOMP achieves the best performance. LASSO exhibits the highest Test MAE, which can be attributed to the estimation bias introduced by the shrinkage effect of the non-negativity constraint on the coefficients. NNOMP performs better, but still ranks as the second worst. This is primarily because the atom selection process is dominated by the large-magnitude columns of the measurement matrix, rendering the correlation with the residual largely ineffective. SWOMP, which aims to distinguish closely parallel columns using second-order statistics between the RSRP and channel APS, exhibits poorer performance due to inaccurate statistics estimation caused by a large number of missing beam measurements. WNOMP uses a more advanced weighted atom selection metric that combines the correlation between normalized columns and the residual with the column magnitudes, while maintaining the probability of selecting large-magnitude columns. However, its assumption that channel paths are more likely to exist at angles with large beam gains is not fundamental, as channel paths are entirely determined by the propagation environment. By incorporating environmental priors through geometry-aware weighted atom selection, GM-NNOMP is more likely to recover the correct channel paths. In addition, the missing-value-aware constrained NNLS enhances robustness to missing beams. 

\subsection{Performance Analysis of Joint Grid Construction and Channel APS Estimation}

To evaluate the effectiveness of the joint grid construction and channel APS estimation, we compare the following three baseline methods with Algorithm~\ref{alg:JGCCAE}.

\textit{1) Baseline 1:} This baseline method first adopts the naive uniform grids, then employs WNOMP to estimate the channel APS for each grid.

\textit{2) Baseline 2:} Similar to Baseline 1, except that this method employs $k$-means clustering to construct grids using locations.

\textit{3) Baseline 3:} In this baseline, we employ $k$-means clustering to construct grids based on multi-beam RSRP\footnote{For multi-beam RSRP with missing values, we perform clustering based on the valid distance in \eqref{eq:valid_d}.}. Then, the channel APS of each grid is estimated using WNOMP.

To provide a comprehensive analysis, we compare the Train MAE and test MAE of different methods under varying numbers of grids $K$\footnote{For uniform grids, the grid width increases uniformly from $10~m$ to $55~m$.}, where the Train MAE is computed on the training set as follows
\begin{align}
    \text{Train MAE} = \frac{1}{K}\sum_{k=1}^K\frac{1}{M\times|\mathcal{G}_k|}\sum_{i \in \mathcal{G}_k}\| \rmA\hat{\rvx}_k - \rvy_i\|_1,
\end{align}
where $\rmA$ denotes the measurement matrix before network parameter adjustments, $\mathcal{G}_k$ denotes the index set of MR samples within the grid $k$ in the training set, and $\rvy_i$ denotes the multi-beam RSRP of the $i$-th MR sample in the training set.

To eliminate the impact of localization errors, all methods employing location information for grid construction are based on true locations. As shown in Fig.~\ref{fig:train_mae} and Fig.~\ref{fig:test_mae}, we present the Train MAE and Test MAE of different methods as the number of grids $K$ increases, where the regularization factor $\beta$ and the maximum iteration number for MR-LSCM are set to 1 and 15, respectively. Baseline 1 and Baseline 2, which construct grids based solely on geographic location, exhibit the highest Train MAE and consequently result in poor Test MAE performance. This is because the RSRP variation within each grid is significant, indicating inconsistent channel characteristics that cannot be accurately represented by a unified channel APS. As shown in Fig.~\ref{fig:grid1} and Fig.~\ref{fig:grid2}, we present the grid construction results based on Baseline 1 and Baseline 2. It can be observed that grid construction based solely on geographic location fails to capture the spatial distribution of RSRP and does not align well with the localized propagation environment. 

\begin{figure}[t]
    \centering
    \includegraphics[width=0.95\linewidth]{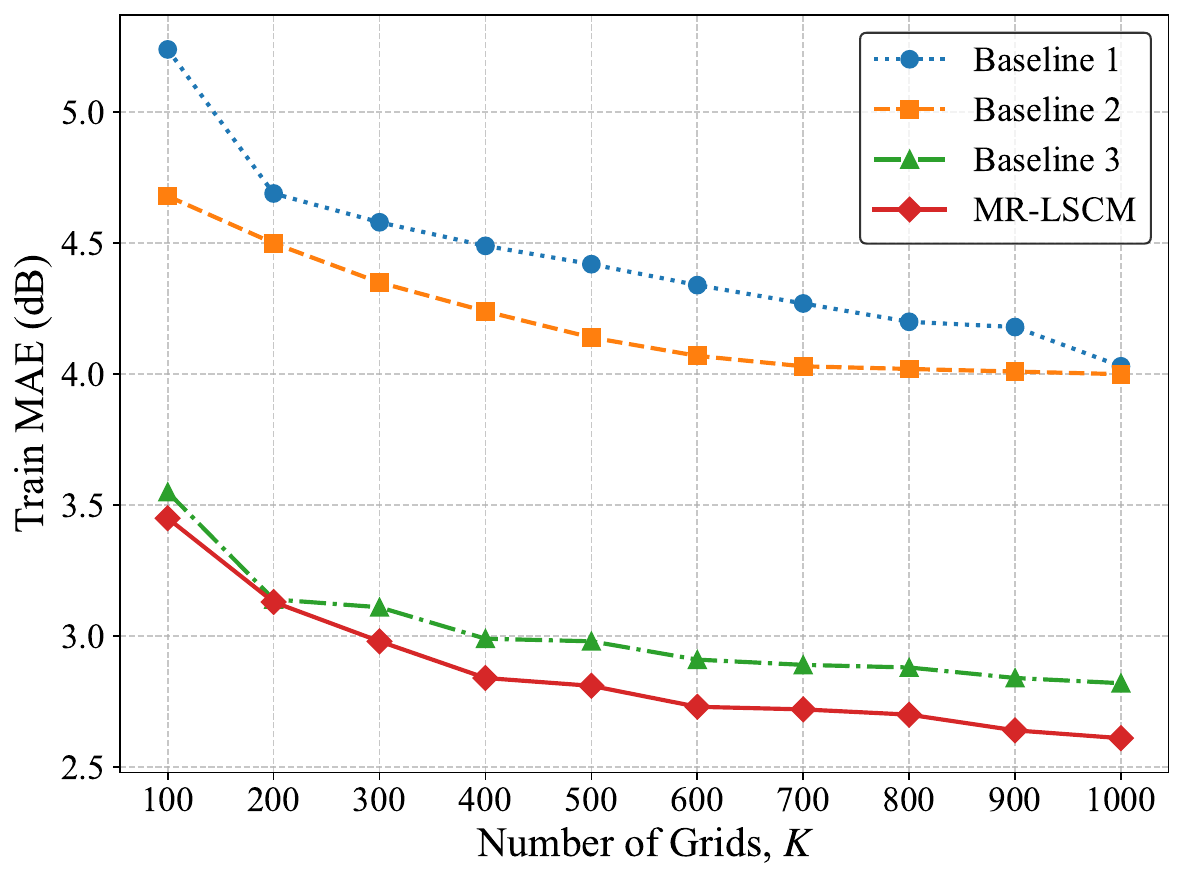}
    \caption{Train MAE comparisons for different baselines and MR-LSCM under varying numbers of grids $K$.}
    \label{fig:train_mae}
    \vspace{-5pt}
\end{figure}

In contrast, the other methods, which incorporate multi-beam RSRP in grid construction, achieve lower Train MAE. However, Baseline 3 exhibits poor Test MAE performance, which can be attributed to the underdetermined nature of the channel model. Considering the statistical relationship between multi-beam RSRP and channel APS (i.e., $\rvy = \rmA\rvx$, where $\rvy \in \mathbb{R}^{8 \times 1}$ and $\rvx \in \mathbb{R}^{6552 \times 1}$), the same multi-beam RSRP may originate from different channel APS. As shown in Fig.~\ref{fig:grid3}, the grids constructed solely based on RSRP are highly scattered in geographical space, which may lead to samples from regions with different channel characteristics being assigned to the same grid. Therefore, a lower Train MAE does not necessarily equate to consistent channel characteristics within each grid. MR-LSCM, which construct grids based on joint features of predicted position and multi-beam RSRP, achieves lowest Test MAE by reducing both RSRP variation and location differences within each grid, as shown in Fig.~\ref{fig:grid4}. This demonstrates that our objective function in \eqref{eq:MR-LSCM} can facilitate grid construction that better aligns with the complex localized propagation environment. Moreover, Fig.~\ref{fig:test_mae_joint} presents the Test MAE for MR-LSCM at each iteration of Algorithm~\ref{alg:JGCCAE}, demonstrating progressively improving performance. This validates our joint grid construction and channel APS estimation design, where the mutually refining mechanism between these two tasks enables superior solutions under complex environments and spatially non-uniform data.

\begin{figure}[t]
    \centering
    \includegraphics[width=0.95\linewidth]{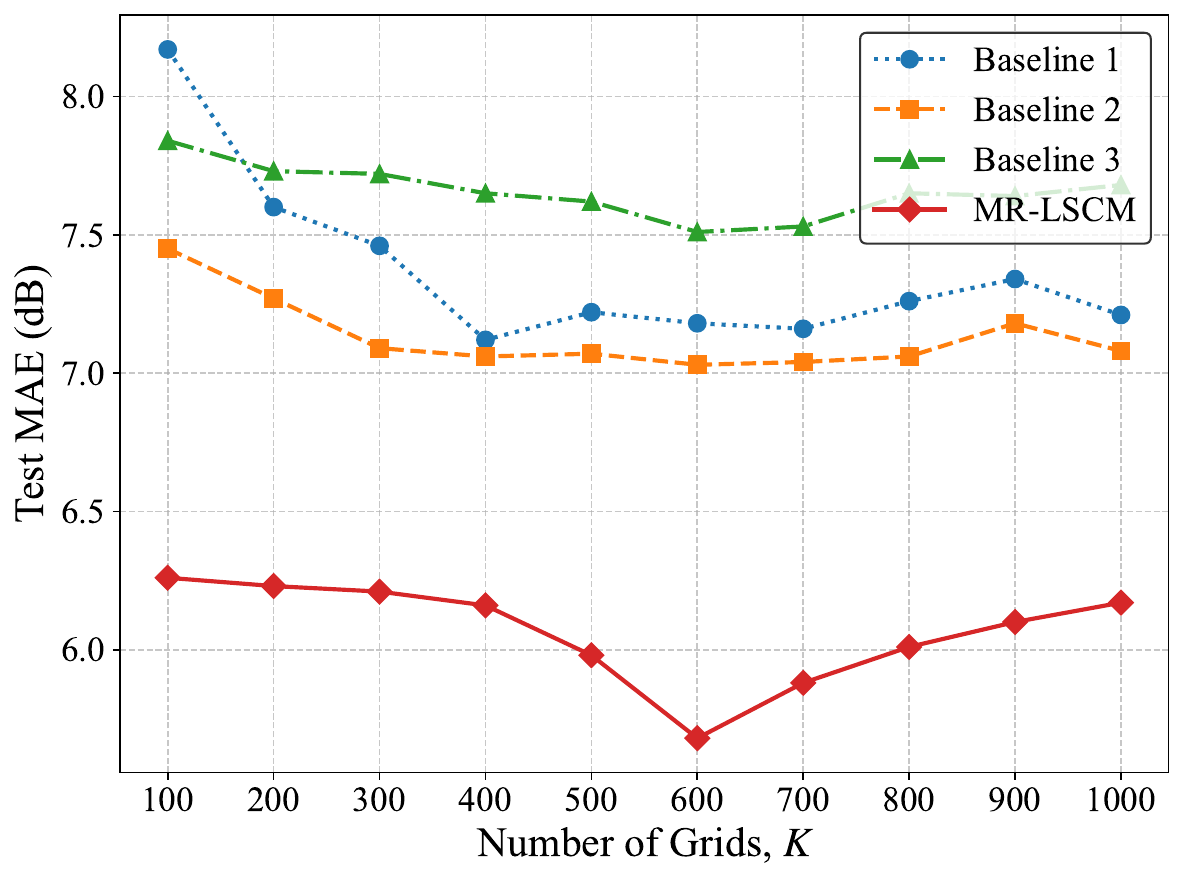}
    \caption{Test MAE comparisons for different baselines and MR-LSCM under varying numbers of grids $K$.}
    \label{fig:test_mae}
    \vspace{-5pt}
\end{figure}

\begin{figure}[t]
\centering
    \resizebox{0.95\columnwidth}{!}{
        \subfloat[Baseline 1\label{fig:grid1}]{
            \includegraphics[height=5cm]{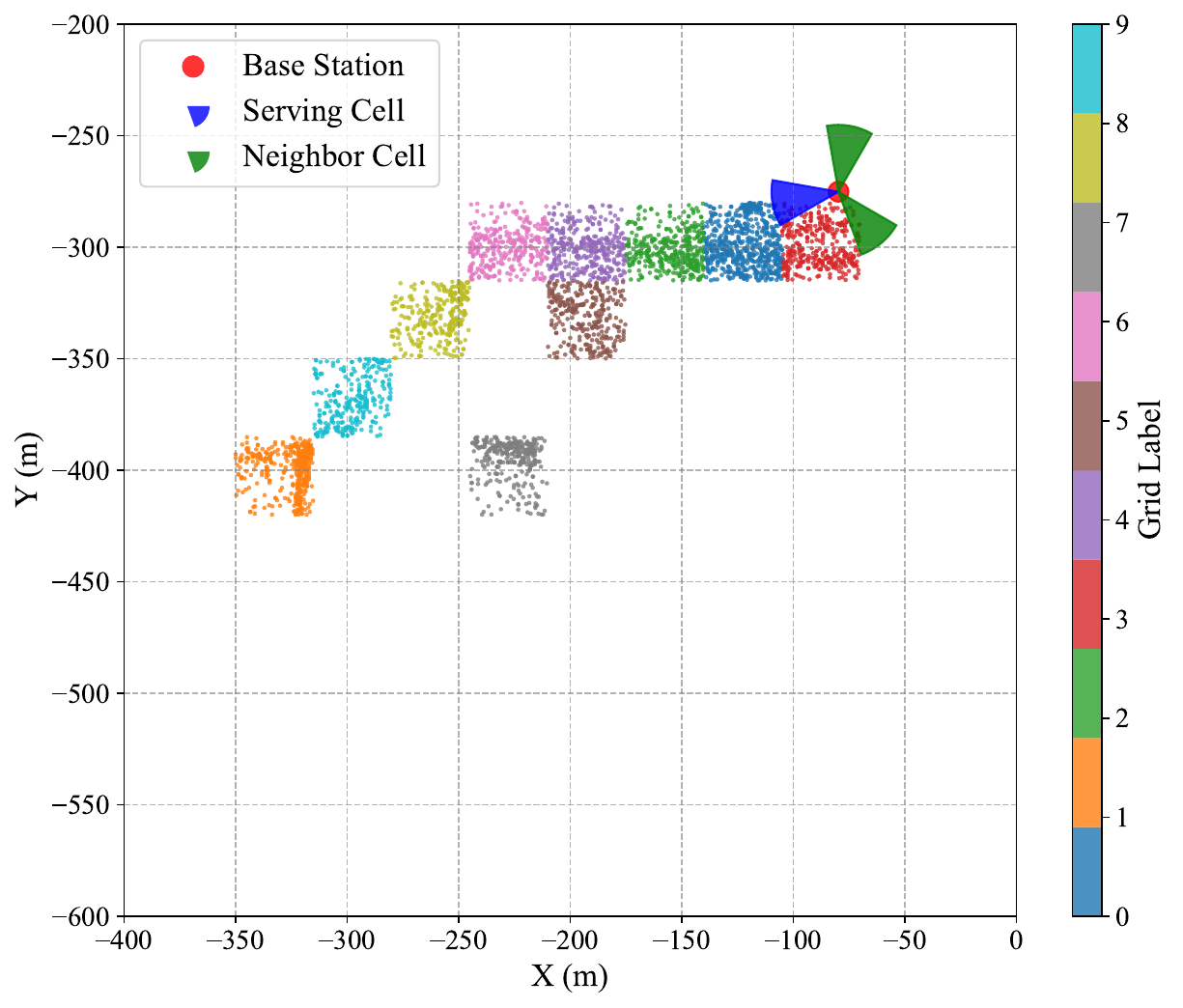}
        }
        \subfloat[Baseline 2\label{fig:grid2}]{
            \includegraphics[height=5cm]{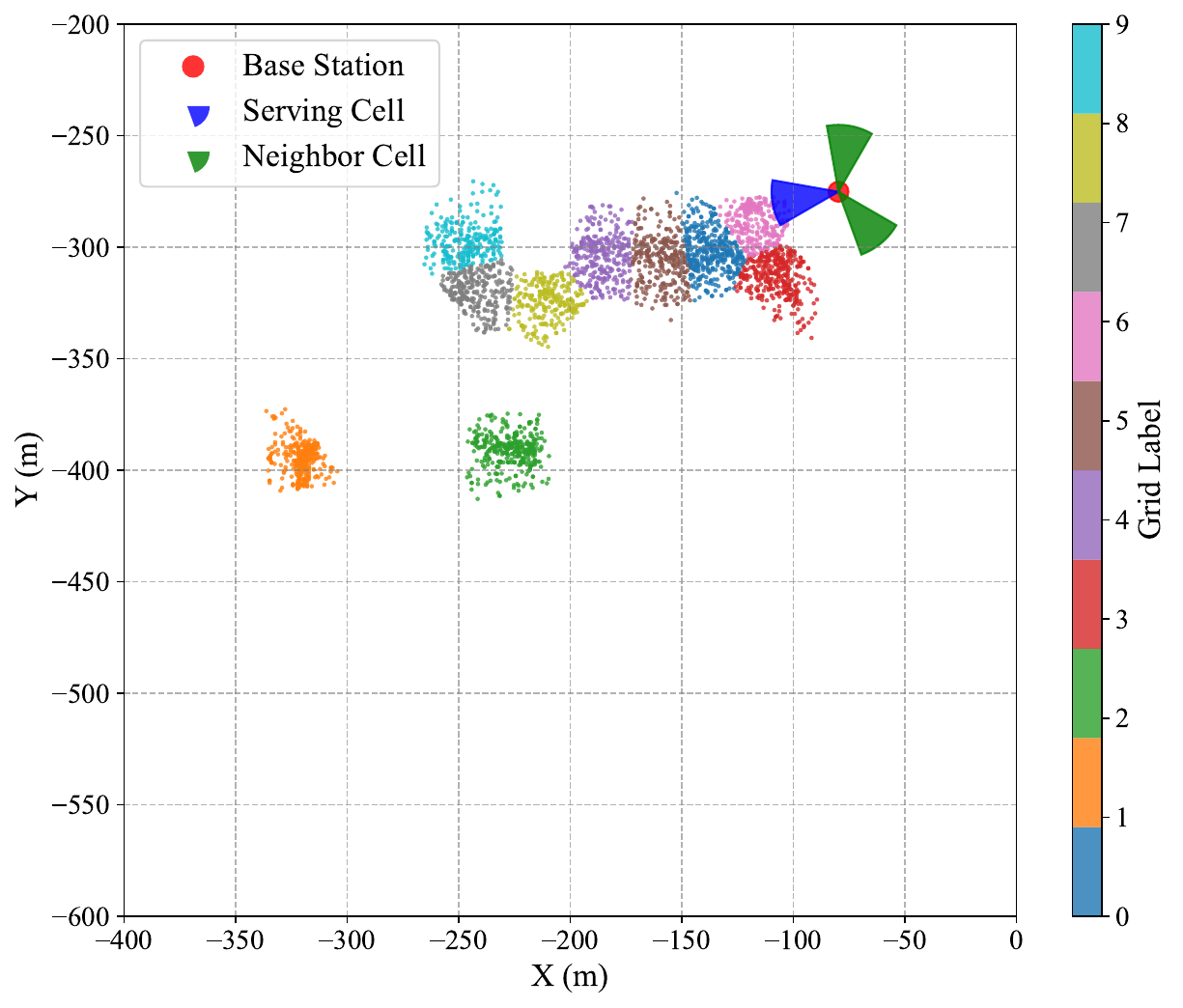}
        }
    }
\\
\centering
    \resizebox{0.95\columnwidth}{!}{
        \subfloat[Baseline 3\label{fig:grid3}]{
            \includegraphics[height=5cm]{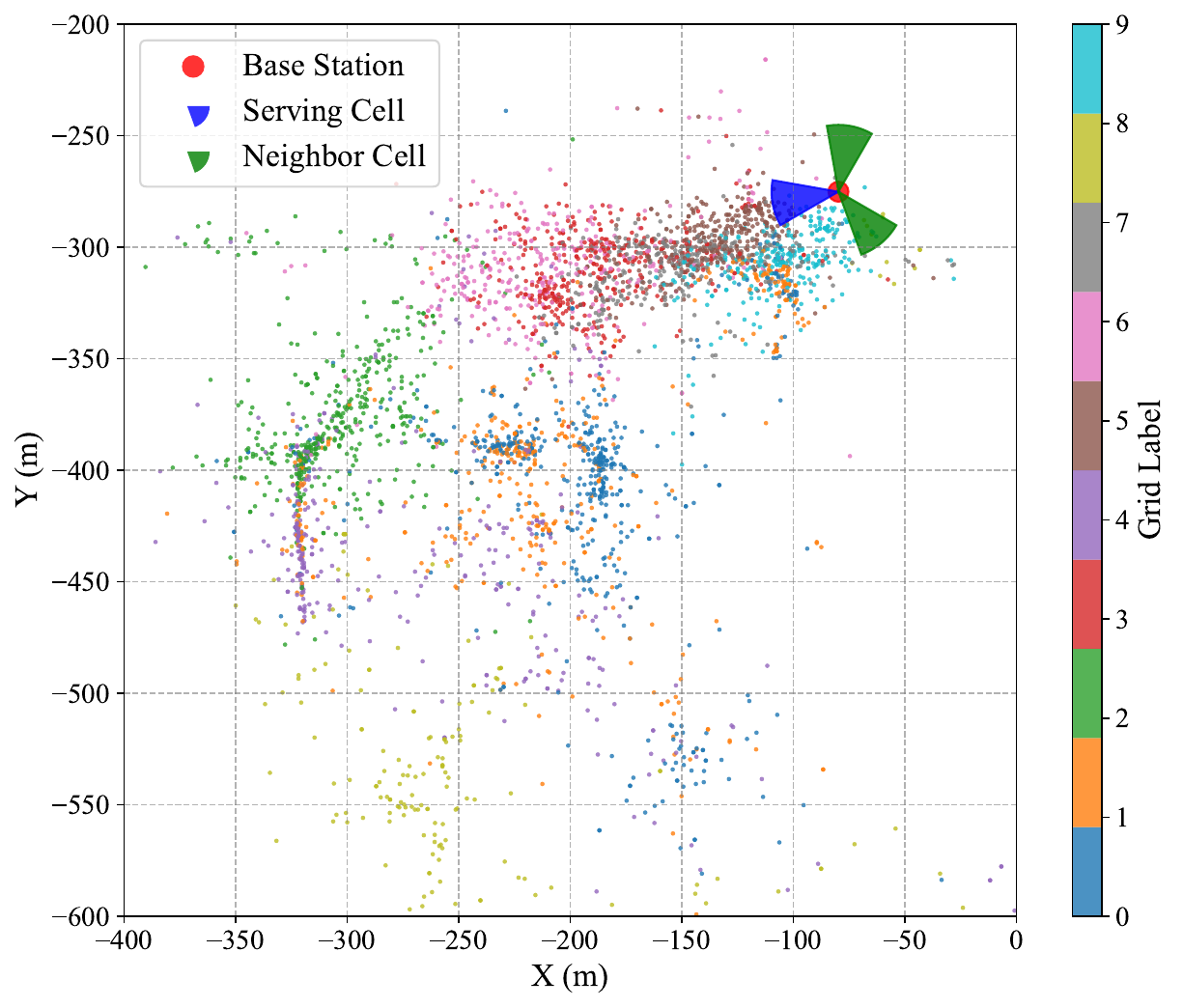}
        }
        \subfloat[MR-LSCM\label{fig:grid4}]{
            \includegraphics[height=5cm]{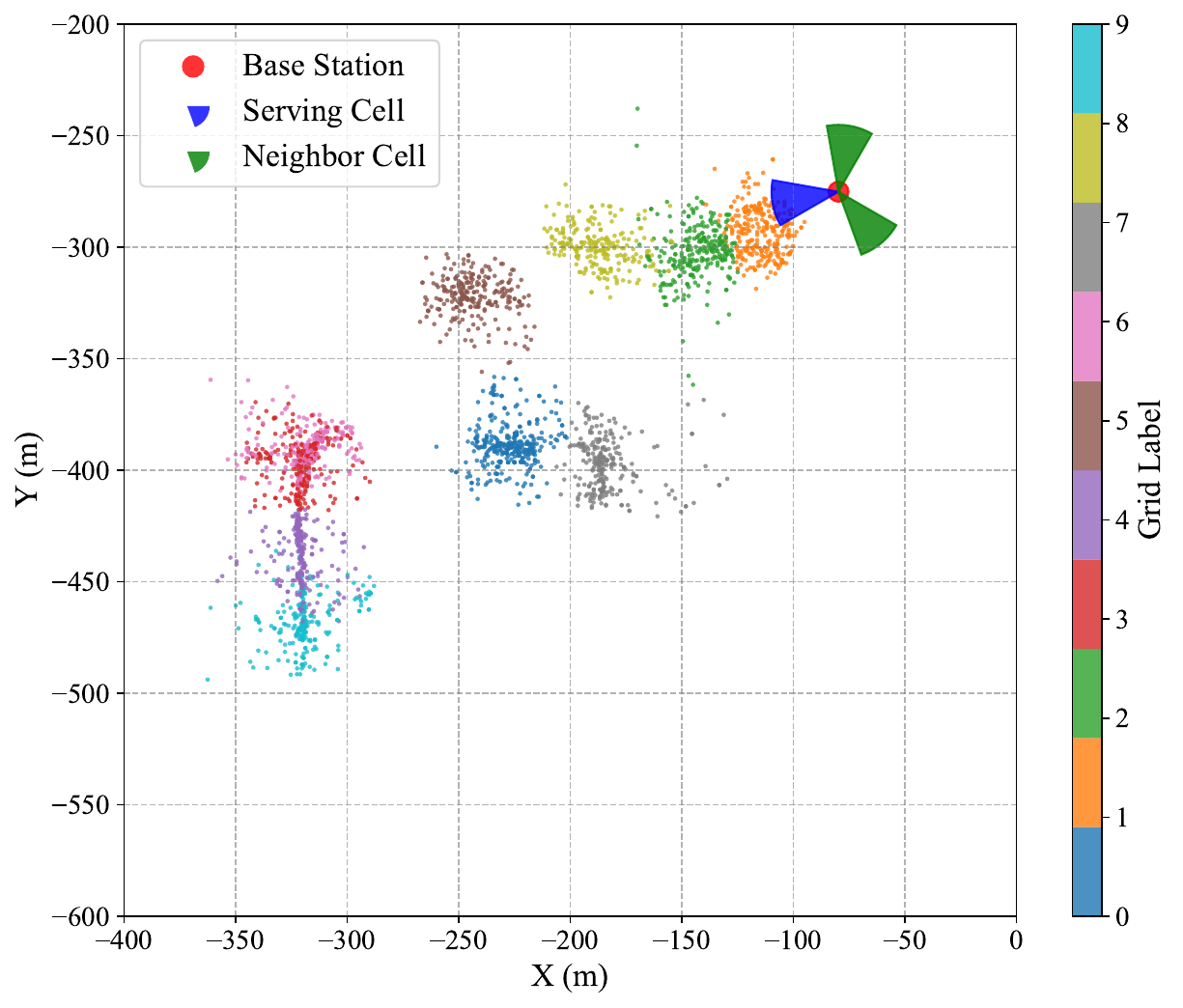}
        }
    }
\caption{Grid construction based on different methods. The number of grids $K$ is set to 500, and for clarity, only the top 10 grids with the largest number of samples are displayed.}
\vspace{-5pt}
\end{figure}

\begin{figure}[t]
    \centering
    \includegraphics[width=0.95\columnwidth]{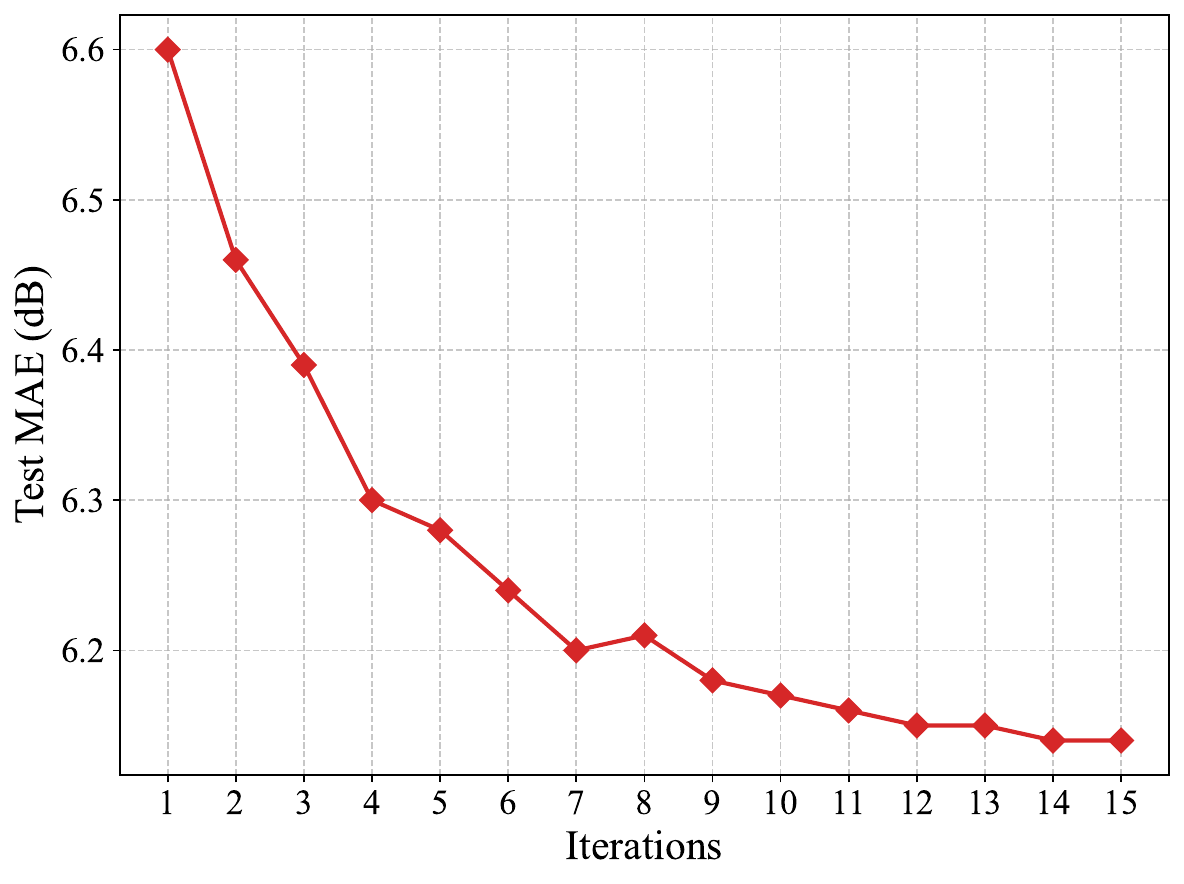}
    \caption{Convergence trend of Test MAE for MR-LSCM. The number of grids $K$ is set to 1000.}
    \label{fig:test_mae_joint}
    \vspace{-5pt}
\end{figure}

We also observe that the optimal number of grids for different methods follows a similar trend: as the number of grids increases, the test MAE initially decreases and then increases. For example, MR-LSCM achieves the lowest test MAE of 5.68 dB when the number of grids is 600. This is because a larger number of grids enables finer discretization but also exacerbates the effects of sparse measurements. In addition, increasing the number of grids leads to higher storage and computational overhead, which should also be considered when determining the appropriate grid number.

\subsection{Robustness Analysis of MR-LSCM Framework}
In this experiment, we focus on analyzing the robustness of the overall framework against localization errors. Specifically, we train the HGNN-Loc model with varying proportions of location labels followed by jointly constructing grids and estimating channel APS based on the predicted locations. When 100\% locations labels are used, the results correspond to those in the previous subsection. We compare the performance of MR-LSCM with Baseline 1 and Baseline 2. The number of grids is set to 1000. Fig.~\ref{fig:test_mae_loc} displays the Test MAE of different methods. Baseline 1 and Baseline 2, which relies solely on location information, exhibits significant performance fluctuations and is more sensitive to localization errors. In contrast, MR-LSCM, which accounts for the RSRP variance within grids, achieve more consistent performance. This is because the spatial continuity of RSRP mitigates the impact of localization errors on grid construction. Overall, MR-LSCM demonstrates the best performance and robustness.

\begin{figure}[!t]
    \centering
    \includegraphics[width=0.95\columnwidth]{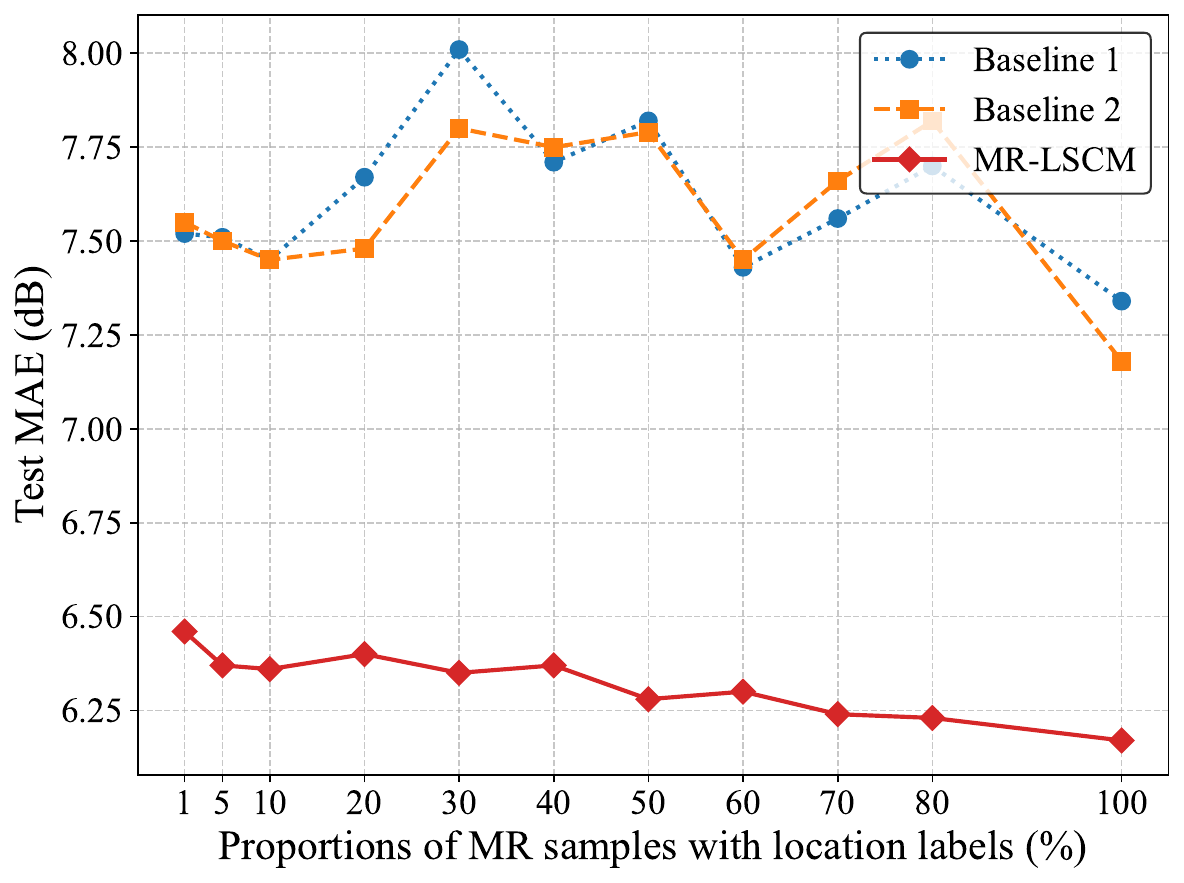}
    \caption{Test MAE comparisons for different methods under varying proportions of MR data with location labels.}
    \label{fig:test_mae_loc}
    \vspace{-5pt}
\end{figure}

\section{Conclusion}
\label{Section:Conclusion}
In this paper, we propose the MR-LSCM framework to address the limitations of DT-based LSCM by exploiting the low-cost and extensive collection of MR data. To enhance LSCM with MR data, we design two modules tailored to its characteristics. The first module employs the proposed HGNN-Loc to address the lack of device location information in MR data. In HGNN-Loc, MR data are modeled as a distance-aware hypergraph to capture spatial correlations among samples by utilizing their inherent multi-modal information. The locations are subsequently extracted through hypergraph convolution. Through semi-supervised learning, HGNN-Loc is effectively trained with minimal location labels. To enhance MR-LSCM in complex environments with spatially non-uniform data, the second module integrates grid construction and channel APS estimation by leveraging their correlation. We formulate a joint clustering and sparse recovery problem and solve it via alternating optimization. Grid construction is achieved by clustering based on predicted locations and estimated channel APS. The channel APS of each grid is estimated via sparse recovery based on average multi-beam RSRP. To improve the robustness of sparse recovery under the ill-conditioned measurement matrix and incomplete observations, we introduce a GM-NNOMP algorithm that incorporates environmental information and physical priors. Through extensive experiments on a real-world MR dataset, we demonstrate the superior performance and robustness of MR-LSCM in localization and channel modeling, ensuring reliable digital twin-assisted network optimization.

\bibliographystyle{IEEEtran}
\bibliography{IEEEabrv,main}

\end{document}